%% file: DiDwContinuousMediators.tex
\documentclass[12pt,reqno]{amsart}
\usepackage{graphics,amsmath,amssymb,epsfig,stmaryrd,color,amsaddr}
\usepackage{subfig,amsfonts,geometry,array}
\usepackage{graphicx}
\graphicspath{ {./images/} }
\usepackage{amsthm}
\allowdisplaybreaks
\usepackage{mathabx} 
\usepackage{flafter}
\usepackage[svgnames]{xcolor}
\usepackage{bbm}
\usepackage{booktabs}
\usepackage{hyperref}
\definecolor{darkblue}{rgb}{0.0,0.0,0.3}
\hypersetup{colorlinks,breaklinks,linkcolor=darkblue,urlcolor=darkblue,anchorcolor=darkblue,citecolor=darkblue}
\newcommand{\indep}{\rotatebox[origin=c]{90}{$\models$}}
\usepackage{natbib}
\usepackage[normalem]{ulem}
\usepackage{verbatim}
\usepackage{multirow}
\usepackage{bm}
\usepackage{enumerate}
\usepackage{dsfont}
\usepackage{pdflscape}
\usepackage{booktabs,threeparttable}
\newcommand{\sym}[1]{\ifmmode^{#1}\else\(^{#1}\)\fi}

\theoremstyle{plain}

\newtheorem{claim}{Claim}
\newtheorem*{assumption*}{Assumption}
\newtheorem{assumption}{Assumption}

\newtheorem{definition}{Definition}
\newtheorem{example}{Example}

\newtheorem{lemma}{Lemma}

\newtheorem{proposition}{Proposition}
\newtheorem{remark}{Remark}

\numberwithin{equation}{section}

\usepackage{pgfplots} 
\pgfplotsset{compat=newest}
\usepgfplotslibrary{fillbetween} 
\usepgfplotslibrary{polar} \usepgflibrary{shapes.geometric}  
\usetikzlibrary{calc}  
\pgfplotsset{my style/.append style={axis x line=middle, axis y line= middle, xlabel={$\tau^{DiD}$}, ylabel={$ATT$}, 
		every axis x label/.style={
			at={(ticklabel* cs:1.00)},
			anchor=west,
		},
		every axis y label/.style={
			at={(ticklabel* cs:1.00)},
			anchor=south,
		},
		axis equal}}

\begin{document}
\begin{titlepage}
\nonstopmode

\title[DiD in the Presence of Time-Varying Mediators and Lagged outcomes]{Difference-in-Differences Models in the Presence of Time-Varying Mediators and Lagged outcomes}
\author{Kyunghoon Ban, Zhengrun Chen, and D\'esir\'e K\'edagni}
\address{Rochester Institute of Technology and UNC-Chapel Hill} 
\noindent \date{\scriptsize{The present version is as of \today. We thank Bocar Ba, St\'ephane Bonhomme, Jane Fruehwirth, Martin Huber, Valentin Verdier, and participants at the UNC econometrics workshop, the 2025 SEA Meeting, the 2026 AfES conference, seminar participants at Rensselaer Polytechnic Institute for helpful comments. This draft supersedes a previous version titled ``\textit{Difference-in-Differences Models in the Presence of Time-Varying Mediators}.''
All errors are ours.
Kyunghoon Ban (kban@saunders.rit.edu); 
Zhengrun Chen (zhengrun@unc.edu);
D\'esir\'e K\'edagni (dkedagni@unc.edu)}}

\begin{abstract}
{
We study difference-in-differences (DiD) designs in which a binary treatment changes an endogenous time-varying (continuous, discrete, or mixed) mediator that in turn affects an outcome. We allow for the inclusion of lagged outcomes in the model. Under our model assumptions, we show that the usual DiD estimand mixes the average direct effect on the treated, the average indirect effect, and a trend bias term. A two-way fixed effects (TWFE) regression that controls for the mediator does not recover the average direct treatment effect on the treated. We show that a DiD estimand conditional on the observed mediator path identifies the conditional average direct effect for treated units at that path, and that averaging over the treated path distribution identifies the average direct effect even when unconditional parallel trends fails. A stable average mediator effect assumption helps recover the average mediator and indirect effects. The framework extends to multivariate mediators, nonlinear DiD, and multiple-treatment-period settings. Existing doubly robust estimators can be used to conduct inference. Revisiting the effects of railroad access on agricultural land values, the specification yields a positive direct component not mediated by measured market access, while the corresponding indirect component is small and imprecise. 
}
\end{abstract}
\maketitle

{\footnotesize \textbf{Keywords}: Difference-in-differences, endogenous time-varying covariates, causal mediation.

}

\end{titlepage}

\section{Introduction}

{
The use of difference-in-differences (DiD) models in economics and social sciences has increased rapidly in recent years. The method relies on a key assumption called parallel trends. The presence of mediators and/or lagged outcomes in the model could compromise the validity of this assumption. In this paper, we study identification in DiD models when mediators and lagged outcomes are present. 
We study this problem for a binary treatment and a mediator that may be continuous, discrete, or mixed, may vary over time, and may be related to unobserved determinants of the outcome.

More specifically, to provide some intuition behind our main research question of interest, consider the following model
\begin{equation*}
\left\{
\begin{array}{lcl}
Y_{it}&=&\theta_{it}D_{it}+\rho Y_{it-1}+\beta_{it}M_{it}+\lambda_i+\eta_t+U_{it},\\[0.2cm]
M_{it}&=&\psi_t(D_{it},\gamma_i,\tilde{V}_{it}),\\[0.2cm]
D_{it}&=&h_t(\lambda_i,\gamma_i,\varepsilon_{it}),
\end{array}
\right.
\end{equation*}
where the vector $(\{Y_{it},D_{it},M_{it}\}_{i\in\mathcal N, t\in \mathcal T})$ is observed, with $\mathcal N$ and $\mathcal T$ denoting the sets of individuals and time periods, respectively, while everything else is unobserved/unknown. 
In the specification above, $Y_{it}$ is the outcome, $M_{it}$ is the mediator, and $D_{it}$ is the treatment.\footnote{Throughout the paper, we assume that no individual receives the treatment in the baseline period, i.e., $D_{it}=0$ a.s. for $t\leq 0$. So, in the canonical DiD setting, our focus is on the treatment $D_{i1}=D_i$.}  The outcome equation contains a random direct effect of treatment $\theta_{it}$, a random effect of the mediator on the outcome $\beta_{it}$, individual fixed effects $\lambda_i$, time effects $\eta_t$, and idiosyncratic shocks $U_{it}$.  The mediator, which can be continuous, discrete or mixed, is specified as an unknown function $\psi_t$ of a vector of unobservables $\tilde{V}_{it}$, its own individual fixed effect $\gamma_i$, where the function $\psi_t$ can change arbitrarily over time.  Treatment selection is allowed to depend on the fixed effects in both the outcome and mediator equations and on an additional unobservable random shock $\varepsilon_{it}$, through an unknown function $h_t$.  The outcome and mediator fixed effects may be arbitrarily dependent (e.g., a special case is common fixed effects, $\lambda_i=\gamma_i$).  The direct effect is heterogeneous and may depend on the mediator path. The mediator effect on the outcome may also vary across individuals and over time. The model allows for feedback by including a lagged outcome $Y_{it-1}$.  
In this paper, we are interested in measuring the direct effect of treatment $D_{it}$ on the outcome $Y_{it}$, that is, we want to learn about some features of $\theta_{it}$. 

There are many empirical examples where there is a mediator between the treatment and the outcome of interest. For example, an expansion of health insurance (Medicare or Medicaid) can affect health conditions both directly and indirectly through preventive care. Also, a school reform such as an introduction of a new curriculum can affect wages directly and indirectly through educational attainment, or a local regulation may affect productivity directly and indirectly through asset values, input prices, or labor-market conditions. In these examples, the mediator is not an ordinary pre-treatment covariate but a post-treatment object whose distribution may be changed by the policy and affected by unobserved determinants of the outcome of interest. Consequently, simply controlling for it in a two-way fixed effects regression will generally induce some bias in the attempt to isolate the direct effect.

To illustrate the main idea behind our identification approach, we start from the canonical two-period, two-group DiD design. We first highlight the fact that the standard parallel trends assumption for the untreated outcome may not hold under our proposed framework and the set of assumptions we consider. Even when the idiosyncratic shocks satisfy independence restrictions, treated and untreated units can have different untreated outcome trends because they have different mediator fixed effects and because the average return to the mediator may change over time. Second, we emphasize that a TWFE regression controlling for the mediator generally fails to recover the average direct treatment effect on the treated because residualizing the treatment with respect to the mediator mechanically changes the estimand. Hence, the resulting coefficient generally mixes the direct effect with terms that depend on mediator effect heterogeneity and treatment-induced mediator changes.

Afterwards, we carefully propose a different route. We consider the difference between outcome changes for treated and untreated units, conditional on the full observed mediator path $(M_{i0},M_{i1})$. Then, we show that, under a DiD analogue of the sequential ignorability assumption considered in \cite{imai2010identification}, this path-conditional DiD estimand identifies the conditional average direct treatment effect on the treated, denoted $CADTT$. Averaging this conditional estimand over the mediator path distribution of treated units identifies the average direct treatment effect on the treated, $ADTT$.  Importantly, this result does not require the standard unconditional parallel trends assumption and is not a standard conditional parallel trends argument with a post-treatment covariate. 

Furthermore, we clarify that the standard DiD estimand can be decomposed into average direct and indirect effects plus a bias term that vanishes when the standard parallel trends assumption holds, which can be implied when the average effect of the mediator on the outcome is stable across periods under our framework. Since the direct effect is identified by the path-conditional DiD approach, the average indirect effect can consequently be identified. Accordingly, the average mediator effect on the outcome of interest can be identified as well when the treatment effect on the mediator is identified from a DiD estimand for the mediator.

We extend the model by allowing for the inclusion of a lagged outcome. The inclusion of past outcomes in DiD models alone is another contribution of this paper. See a recent paper by \cite{bonhomme2026back} for a review on issues related to the inclusion of past outcomes in panel data models with heterogeneity. The intuition behind the identification approach is the following. Suppose that there are no mediators and there are three pre-treatment periods $t \in \{-2,-1,0\}$ followed by one post-treatment period $t=1$. We can write $Y_{i0}-Y_{i(-1)}=\rho (Y_{i(-1)}-Y_{i(-2)})+\eta_0-\eta_{(-1)}+U_{i0}-U_{i(-1)}.$ Assuming independence between $(\lambda_i,\varepsilon_{i1})$ and $(U_{i(-1)},U_{i0},U_{i1})$, we have $D_{i1}$ independent of $U_{i0}-U_{i(-1)}$. On the other hand, the model implies $Y_{i(-1)}-Y_{i(-2)}=(\rho-1)Y_{i(-2)}+\lambda_i+\eta_{(-1)}+U_{i(-1)},$ which likely implies $Cov(Y_{i(-1)}-Y_{i(-2)},D_{i1})\neq 0$. Hence, we can use the treatment $D_{i1}$ as an instrumental variable (IV) for $Y_{i(-1)}-Y_{i(-2)}$ and recover $\rho$ as the IV estimand from the regression of $Y_{i0}-Y_{i(-1)}$ on $Y_{i(-1)}-Y_{i(-2)}$. We combine this identification strategy with the one developed for our mediation analysis to identify causal effects in the above model.  

The mediator path argument also extends to multiple treatment periods \citep{callaway2021difference, dechaisemartin2020two}, staggered adoption as well as nonstaggered adoption. 
For estimation, the identified $ADTT$ can be written as the integral of a conditional DiD estimand, where the conditioning variables are the pre- and post-treatment mediator values, with respect to the observed mediator path distribution for the treated group.  We therefore use the doubly robust (DR) DiD estimator of \cite{sant2020doubly}, applied with the mediator path as the conditioning vector.

This paper is related to several branches of the literature. First, it contributes to the causal mediation analysis. Classical mediation analysis studies decompositions of total effects into direct and indirect components and highlights that mediator endogeneity is a fundamental obstacle \citep{robins1992identifiability, pearl2001direct, imai2010identification}. Identification has been based on selection-on-observables assumptions, with the literature developing sensitivity analyses, nonparametric identification results, decompositions allowing interaction between treatment and the mediator, and machine-learning-based estimators \citep{vanderweele2014unification,celli2022causal,farbmacher2022causal}. In contrast, we allow selection into treatment through unobserved fixed effects and exploit within-unit changes through the DiD structure, while explicitly allowing the mediator to be time-varying, endogenous, and affected by the treatment.

Second, our paper is more closely related to recent work on mediation in DiD designs. \cite{deuchert2019direct} propose a DiD approach for direct and indirect effects with a binary mediator in principal strata. \cite{hsia2025causal} develop mediation analysis for DiD and panel data under common-trend-type assumptions and linear/additive specifications. More relatedly, \cite{schenk2024mediation} studies direct and indirect effects in DiD designs in the presence of a binary time-varying mediator. His identification approach relies on a set of parallel trends assumptions for the outcome and mediator, and mean independence of the treatment with respect to the untreated mediator. Our framework, on the other hand, allows for continuous, discrete, or mixed mediators. We also allow for dependence between the potential mediators and the treatment. Furthermore, \cite{huber2026difference} propose a DiD mediation framework for multivalued discrete or continuous treatments and mediators using conditional parallel trends and double/debiased machine learning. In their setting, the mediator is time-invariant, while it is time-varying in our framework.

Third, our analysis contributes to work on treatment-affected covariates in panel treatment effect models. Exogenous time-varying covariates can make parallel trends more plausible, but including endogenous time-varying covariates as controls in a TWFE regression can alter the TWFE estimand and introduce some bias when the treatment changes those covariates. Notably, \cite{caetano2024difference} formalizes several problems with TWFE regressions under conditional parallel trends, including a hidden linearity bias. \cite{brown2026direct} develop a panel approach with interactive fixed effects that permits treatment-affected covariates and separates the part of the estimated treatment effect due to the covariates from the part that is not. Our approach shares the view that treatment-affected variables should not be handled as ordinary controls. The difference is that we explicitly interpret such a variable as a mediator and identify direct and indirect effects for treated units.

Finally, our paper connects to extensions of DiD beyond the mean and linear models. Our nonlinear extension is in the spirit of changes-in-changes and nonlinear DiD approaches that identify counterfactual outcome distributions under restrictions on latent shocks rather than mean parallel trends \citep{athey2006identification,wooldridge2023simple,kim2025difference}. In our setting, the nonlinear argument is applied conditionally on the mediator path, which identifies distributional and quantile versions of direct effects. We also extend the baseline framework to multivariate mediators.

The rest of the paper proceeds as follows. Section~\ref{anaF} sets out the baseline model and causal parameters. Section~\ref{sec:baseline-id} derives the conventional estimand diagnostics and the mediator-path identification result, and Section~\ref{sec:po-framework} restates that result in potential outcome notation. Section~\ref{sec:multiple} generalizes our approach to multiple treatment periods, while Section \ref{sec:laggedoutcome} extends our results to include a lagged outcome in the model. Section~\ref{sec:extensions} deals with multivariate mediators and nonlinear outcomes. Section \ref{sec:empirical} reports the railroad application, followed by a summary and discussion in Section~\ref{Conclusion}.
}

\section{Baseline structural framework}\label{anaF}
We first consider a canonical two-group, two-period DiD design. Individuals $i$ are observed in a baseline period $t=0$ and a post-treatment period $t=1$. Let $D_i\in\{0,1\}$ denote treatment-group status, $M_{it}$ a continuous mediator, and $Y_{it}$ the outcome. The model is specified as follows:
\begin{eqnarray}\label{seq1}
\left\{ \begin{array}{lcl}
     Y_{it} &=& \theta_i t D_i +\beta_{it} M_{it} + \lambda_i +\eta_t +U_{it},\\ \\
     M_{it} &=& \alpha_{i}t D_i+\gamma_i + \delta_t + V_{it},\\ \\
     D_i &=& h(\lambda_i,\gamma_i,\varepsilon_i),
     \end{array} \right.
\end{eqnarray}
where $t \in\{0,1\},$ the vector $(Y_{i0}, Y_{i1}, M_{i0}, M_{i1}, D_i)$ represents the observed data for each individual $i$, while the vector $\left(\theta_i, \alpha_i, \lambda_i, \gamma_i, \varepsilon_i, \{\beta_{it}, U_{it}, V_{it}\}_{t=0,1}\right)$ is latent. $\eta_t$ and $\delta_t$ are nonrandom time effects, $\lambda_i$ and $\gamma_i$ are time-invariant individual heterogeneity in the outcome and mediator equations, and $U_{it}$ and $V_{it}$ are time-varying shocks. $\theta_i$ is the direct effect of the treatment on the outcome $Y_{i1}$ in period 1, $\alpha_i$ is the effect of the treatment on the mediator $M_{i1}$,  and $\beta_{it}$ is the effect of $M_{it}$ on the outcome $Y_{it}$. Treatment status is determined by the selection equation $D_i=h(\lambda_i,\gamma_i,\varepsilon_i)$, where $h$ is an unknown function and $\varepsilon_i$ is an additional latent unobservable. Selection into treatment therefore depends on both $\lambda_i$ and $\gamma_i$, allowing individuals to sort into treatment on unobserved determinants of both the outcome and the mediator.

In this model, we allow the direct treatment effect to depend flexibly on the mediator path, i.e., $\theta_i\equiv \tilde{\theta}_i(M_{i0},M_{i1})$. The dependence between the fixed effects $\lambda_i$ and $\gamma_i$ is left unrestricted (e.g., $\lambda_i=\gamma_i$ is permitted). Moreover, the model allows the effect of the mediator on the outcome (i.e., $\beta_{it}$) to vary across individuals and over time. 

In the potential outcomes framework, let $M_{it}(d)$ denote the potential mediator under treatment $d$, and let $Y_{it}(d,m)$ be the potential outcome under treatment $d$ and mediator $m$. We write $Y_{it}(d)\equiv Y_{it}(d,M_{it}(d))$ and define $Y_{it}(d,M_{it}(d'))$ as the cross-world potential outcome, where the treatment is set to $d$ and the mediator is held at its value under treatment $d'$. Model \eqref{seq1} then yields
\begin{eqnarray*}
    Y_{it}(d) &\equiv& Y_{it}(d,M_{it}(d))= \theta_i t d + \beta_{it} M_{it}(d) +\lambda_i + \eta_t + U_{it},\\
    Y_{it}(d,M_{it}(d')) &\equiv& \theta_i t d + \beta_{it} M_{it}(d') +\lambda_i + \eta_t + U_{it},\\
           M_{it}(d) &\equiv& \alpha_i t d + \gamma_i + \delta_t + V_{it}. 
\end{eqnarray*}
Note that the above model implies that there are no anticipatory effects of the treatment in the baseline period 0, i.e., $Y_{i0}(0)=Y_{i0}(1)$ and $M_{i0}(0)=M_{i0}(1)$. This no-anticipation restriction is standard in DiD models. Using this potential outcome notation, we define the causal parameters of interest.

\begin{definition} The average total, direct, and indirect treatment effects on the treated in period 1 are defined as 
\begin{align*}
ATT  &\equiv \mathbb E[Y_{i1}(1,M_{i1}(1))-Y_{i1}(0,M_{i1}(0))\mid D_i=1],\\
ADTT(d) &\equiv \mathbb E[Y_{i1}(1,M_{i1}(d))-Y_{i1}(0,M_{i1}(d))\mid D_i=1],\\
AITT(d) &\equiv \mathbb E[Y_{i1}(d,M_{i1}(1))-Y_{i1}(d,M_{i1}(0))\mid D_i=1],
\end{align*}
for $d\in\{0,1\}$.
\end{definition}

These correspond to the natural direct and indirect treatment effects studied in the causal mediation literature  \citep{pearl2001direct,imai2010identification}, restricted to the treated subpopulation. $ADTT(d)$ captures the effect of changing treatment status while holding the mediator at its potential value $M_{i1}(d)$, and $AITT(d)$ is the effect of the treatment-induced change in the mediator while holding treatment fixed at $d$. By construction,
\begin{equation*}
ATT = ADTT(d) + AITT(1-d).
\end{equation*}
Model \eqref{seq1}  is additively separable in the direct and indirect channels: for each $d$, $Y_{i1}(1,M_{i1}(d))-Y_{i1}(0,M_{i1}(d))=\theta_{i}$, and $Y_{i1}(d,M_{i1}(1))-Y_{i1}(d,M_{i1}(0))=\beta_{i1}\alpha_i$.
Thus, the direct effect is invariant to the mediator value at which it is evaluated, while the indirect effect is invariant to the treatment level. Consequently,
\begin{align*}
ADTT(1) & = ADTT(0) = \mathbb{E}[\theta_i\mid D_i=1],\\
AITT(1) & = AITT(0) = \mathbb{E}[\beta_{i1}\alpha_i\mid D_i=1],
\end{align*} and $$ATT =\mathbb{E}[\theta_i\mid D_i=1] + \mathbb{E}[\beta_{i1}\alpha_i\mid D_i=1].$$ 
We therefore suppress the argument $d$ and write $ADTT$ and $AITT$ throughout. This additive separability of the two channels plays an important role in identifying these effects.

\section{Identification: conventional estimands and mediator path DiD}\label{sec:baseline-id}
In this section, we state the main identifying assumptions and then present the baseline identification results.  All identification results hold conditional on exogenous time-invariant covariates vector $X_i$ such that $X_i(0)=X_i(1)=X_i$.

\subsection{Identifying assumptions}
\begin{assumption}[Exogenous fixed effects]\label{ass1p}
    $(\lambda_i, \gamma_i,\varepsilon_i)\ \indep\ (\{V_{it},U_{it},\beta_{it}\}_{t=0,1},\alpha_i)$
\end{assumption}

\begin{assumption}[Conditional independence]\label{ass2p}
    $(\{V_{it}\}_{t=0,1},\alpha_i)\ \indep\ (\{U_{it},\beta_{it}\}_{t=0,1})\mid (\lambda_i,\gamma_i,\varepsilon_i)$
\end{assumption}
Assumption \ref{ass1p} requires the unobservables driving selection into treatment, ($\lambda_i,\gamma_i,\varepsilon_i$), to be independent of all other unobservables in the model except the direct effect $\theta_i$.  By allowing $\theta_i$ to depend on selection unobservables, the model permits selection on the direct gain from treatment.
Assumption \ref{ass2p} further requires that, conditional on the selection unobservables, the latent components in the mediator equation are independent of those in the outcome equation, again except for $\theta_i$. 

Together, Assumptions \ref{ass1p} and \ref{ass2p} can be viewed as a DiD analogue of the sequential ignorability assumption in \cite{imai2010identification}. They imply that, aside from the direct effect $\theta_i$, dependence between the mediator and the outcome arises only through the fixed effects $\lambda_i$ and $\gamma_i$. In particular, they rule out time-varying unobservables that jointly affect the mediator and the outcome. These assumptions yield the implications stated in Lemma \ref{lem:uncond_indep} below.

\begin{lemma}\label{lem:uncond_indep}
    Assumptions \ref{ass1p} and \ref{ass2p} imply the following:
                        $$(\{U_{it},\beta_{it}\}_{t=0,1})\ \indep\ (\{V_{it}\}_{t=0,1},\alpha_i, \lambda_i, \gamma_i,\varepsilon_i)$$
                   \end{lemma}
\begin{proof}
    See Appendix \ref{pf_lem:uncond_indep}.
\end{proof}
Lemma \ref{lem:uncond_indep} is important for proving our main identification results.

\subsection{Why conventional estimands fail}\label{sec:conventional-estimands}

\subsubsection{Standard DiD and the total effect}\label{sec:standard-did-diagnostic}
Let $\Delta Y_i=Y_{i1}-Y_{i0}$ and $\Delta M_i=M_{i1}-M_{i0}$, and define the DiD estimand
$$\tau^{DiD}\equiv \mathbb E[\Delta Y_i\mid D_i=1]-\mathbb E[\Delta Y_i\mid D_i=0]. $$

Under model \eqref{seq1} and Assumptions \ref{ass1p} and \ref{ass2p}, we show that
\begin{equation*}
\tau^{DiD}
=\underbrace{\mathbb E[\theta_i\mid D_i=1]}_{\text{direct}}
+\underbrace{\mathbb E[\beta_{i1}]\,\mathbb E[\alpha_i]}_{\text{indirect}}
+\underbrace{\mathbb E[\beta_{i1}-\beta_{i0}]
\big(\mathbb E[\gamma_i\mid D_i=1]-\mathbb E[\gamma_i\mid D_i=0]\big)}_{\text{parallel-trends bias}}.
\end{equation*}

The last term is due to the violation of the parallel trends assumption,
\begin{align*}
    & \mathbb E[Y_{i1}(0)-Y_{i0}(0)\mid D_i=1]-\mathbb E[Y_{i1}(0)-Y_{i0}(0)\mid D_i=0]\\
=\ & \mathbb E[\beta_{i1}-\beta_{i0}]\big(\mathbb E[\gamma_i\mid D_i=1]-\mathbb E[\gamma_i\mid D_i=0]\big),
\end{align*}
which is nonzero whenever the average mediator effect changes over time and $\gamma_i$ is correlated with treatment status. Hence, $\tau^{DiD}$ generally mixes the direct effect, the indirect effect, and a parallel-trends bias, and has no causal interpretation on its own. When the average mediator effect is stable ($\mathbb E[\beta_{i1}]=\mathbb E[\beta_{i0}]$) or there is no treatment selection on the mediator fixed effects ($\mathbb E[\gamma_i\mid D_i]=\mathbb E[\gamma_i]$), the standard parallel trends assumption holds, the bias vanishes, and the DiD estimand recovers the average total treatment effect on the treated, $\tau^{DiD}=ATT=\mathbb E[\theta_i\mid D_i=1]+\mathbb E[\beta_{i1}]\,\mathbb E[\alpha_i]$.

\subsubsection{TWFE conditioning on $M_{it}$ and the direct effect}\label{sec:twfe-diagnostic}

 A natural question to ask is: does a TWFE estimand that controls for $M_{it}$ identify the direct treatment effect under our model? 
 To address this, we consider the following TWFE regression with a scalar mediator $M_{it}$: for $t\in\{0,1\}$, 
$$Y_{it}=\pi_1 tD_{i}+\pi_2 M_{it}+\phi_{i}+\zeta_{t}+\xi_{it},$$ 
where $\phi_i$ is the individual fixed effect, $\zeta_t$ is the time fixed effect, $\pi_1$ and $\pi_2$ are constant coefficients. In practice, the coefficient $\pi_1$ is often interpreted as an average direct treatment effect on the outcome. However, as shown below, it generally combines the true average direct effect with additional bias terms.

To characterize the TWFE estimands, we introduce the following notation. For any random variables $X$ and $Y$, let $\Pi(Y\mid X)$ denote the linear projection of $Y$ onto $X$, i.e., $\Pi(Y\mid X)\equiv \mathbb{E}(Y)+\frac{\operatorname{Cov}(X,Y)}{\operatorname{Var}(X)}[X-\mathbb{E}(X)]$.  Define the residualized $D_i$ with respect to $\Delta M_i$ as $\widetilde{D}_i\equiv D_i - \Pi(D_i\mid \Delta M_i)$ and the residualized $\Delta M_i$ with respect to $D_i$ as $\widetilde{\Delta M_i} \equiv \Delta M_i - \Pi(\Delta M_i \mid D_i)$.
Then, the TWFE estimands are $\pi_{1}=\frac{\operatorname{Cov}(\widetilde{D}_{i},\Delta Y_{i})}{\operatorname{Var}(\widetilde{D}_{i})}$, and $\pi_2 =  \frac{\operatorname{Cov}(\widetilde{\Delta M_{i}},\Delta Y_{i})}{\operatorname{Var}(\widetilde{\Delta M_{i}})}.$\footnote{Throughout this subsection, assume the relevant second moments are finite and the population rank condition holds: $0<\operatorname{Var}(D_i)<\infty$, $0<\operatorname{Var}(\Delta M_i)<\infty$, $\operatorname{Var}(\widetilde{D}_i)>0$, and $\operatorname{Var}(\widetilde{\Delta M}_{i})>0$.
Equivalently, after partialling out the mediator change, the treatment regressor has nonzero residual variation, and after partialling out treatment status, the mediator change has nonzero residual variation.}
The following claim shows that each of these estimands can be decomposed into three terms.
\begin{claim}\label{claim:interpretTWFE}
Under model \eqref{seq1}, Assumptions \ref{ass1p} and \ref{ass2p}, 
the estimands for $\pi_1$ and $\pi_2$ from the above TWFE regression can be decomposed as:
\begin{align*}
\pi_1
= & \ \mathbb{E}(\theta_{i}\mid D_{i}=1)\\
& +\mathbb{E}(\beta_{i1}-\beta_{i0})
\left(
\frac{\operatorname{Cov}(D_{i},\gamma_{i})}{\operatorname{Var}(\widetilde{D}_{i})}
-
\frac{\mathbb{E}(\alpha_{i})\operatorname{Cov}(\Delta M_{i},M_{i0})}
{\operatorname{Var}(\widetilde{\Delta M}_{i})}
\right) \tag{A} \label{bias:A} \\
& +\frac{\mathbb{E}(\alpha_{i})}{\operatorname{Var}(\widetilde{\Delta M}_{i})}
\big(
\mathbb{E}(\alpha_{i})\operatorname{Cov}(D_{i},\theta_{i}D_{i})
-
\operatorname{Cov}(\Delta M_{i},\theta_{i}D_{i})
\big) \tag{B}\label{bias:B}.
\end{align*}
\begin{align*}
\pi_2
= & \ \mathbb{E}(\beta_{i1})\\
&  + \frac{\mathbb{E}(\beta_{i1}-\beta_{i0})}
{\operatorname{Var}(\widetilde{\Delta M}_{i})}
\Big[
\operatorname{Cov}(\Delta M_{i},M_{i0})
-
\mathbb{E}(\alpha_{i})\operatorname{Cov}(D_{i},\gamma_{i})
\Big]   \tag{C}\label{bias:C} \\
& + \frac{1}{\operatorname{Var}(\widetilde{\Delta M}_{i})}
\Big[
\operatorname{Cov}(\Delta M_i,\theta_{i}D_{i})
-
\mathbb{E}(\alpha_i)\operatorname{Cov}(D_i,\theta_iD_i)
\Big]. \tag{D} \label{bias:D}
\end{align*}
\end{claim}

Claim \ref{claim:interpretTWFE} shows that TWFE generally fails to identify the average direct treatment effect and the average mediator effects. Instead, the coefficients $\pi_1$ and $\pi_2$ combine the causal parameters of interest with several bias terms; see proof in Appendix \ref{pf_claim:interpretTWFE}.

For $\pi_1$, it can be decomposed into the ADTT and two bias terms. The first bias term (\ref{bias:A}) is driven by differences in the average mediator effects on the outcome across periods. This difference would imply that the evolution of untreated potential outcome will depend on both the level of untreated mediator and its change over time:
$$  \Delta Y_i(0) = \beta_{i1}\Delta M_i(0) + \tilde{\beta}_i M_{i0}(0) + \eta_1-\eta_0 + U_{i1}-U_{i0},$$
where $\tilde{\beta}_i = \beta_{i1}-\beta_{i0}$. However, the TWFE specification implicitly imposes that the evolution depends only on the change in mediators. This feature is analogous to the ``hidden linearity bias''\footnote{In \cite{caetano2024difference}, the ``hidden linearity bias'' arises from the fact that the evolution of untreated potential outcomes depends on time-invariant exogenous covariates, the level of time-varying exogenous covariates, and the change in these covariates, which TWFE regressions fail to capture. Instead, we focus on the endogenous time-varying covariates (mediators), but we can flexibly accommodate both time-invariant and time-varying exogenous covariates in our model.} in \cite{caetano2024difference}. However, unlike their setting, our framework does not suffer from the ``explicit linearity bias'', which arises when the conditional expectation of the change in untreated potential outcomes is nonlinear in the change in mediators over time. In our model, this relationship remains linear. In addition, relative to \cite{caetano2024difference}, our setting introduces an additional bias term (\ref{bias:B}) which is driven by the endogeneity of mediators. This bias appears when the average treatment effect on the mediator is nonzero and the direct treatment effects on the outcome, $\theta_i$, is correlated with other unobserved heterogeneity in the model.

For $\pi_2$, it conflates mediator effects on the outcome from both periods, as captured by the term (\ref{bias:C}), and includes an additional bias term (\ref{bias:D}) driven by the dependence between direct treatment effects $\theta_i$ and other unobserved heterogeneity in the model.

These bias terms disappear only under strong conditions. For example, suppose  $\mathbb{E}(\alpha_i)\neq 0$, if the average mediator effect is time-invariant, i.e., $\mathbb{E}(\beta_{i1})=\mathbb{E}(\beta_{i0})$, and the direct treatment effect on the outcome is independent of both the treatment effect on the mediator and the shocks in the mediator equation, i.e., $\theta_{i}\ \indep\ (\alpha_{i},V_{i1},V_{i0})$, which implies $\operatorname{Cov}(\Delta M_i,\theta_{i}D_{i})=\mathbb{E}(\alpha_i)\operatorname{Cov}(D_i,\theta_iD_i)$, then the TWFE regression conditional on $M_{it}$ identifies the average direct treatment effect on the treated $\mathbb{E}[\theta_i\mid D_i=1]$, as well as the average mediator effects in period 1, $\mathbb{E}(\beta_{i1})$.

Figure~\ref{fig:simulation_example} provides a Monte Carlo illustration of these results; see the simulation design in Appendix~\ref{apx:simulation_example}. Under this DGP, the true ATT and ADTT equal $2$ and $1$, respectively.
Nevertheless, standard DiD and TWFE estimators produce biased estimates of these parameters, with possible sign reversals. This motivates the mediator-path identification strategy discussed in Section \ref{sec:path-id}.

\begin{figure}[!htbp]
  \centering
  \begin{minipage}[t]{0.48\textwidth}
    \centering
    \includegraphics[width=\linewidth]{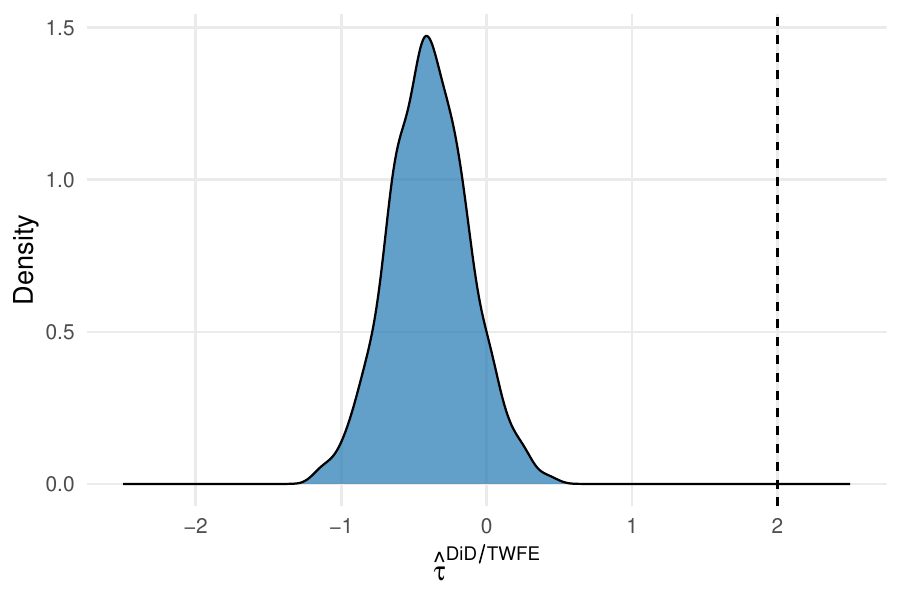}
    {\footnotesize (a) Total effect}
  \end{minipage}\hfill
  \begin{minipage}[t]{0.48\textwidth}
    \centering
    \includegraphics[width=\linewidth]{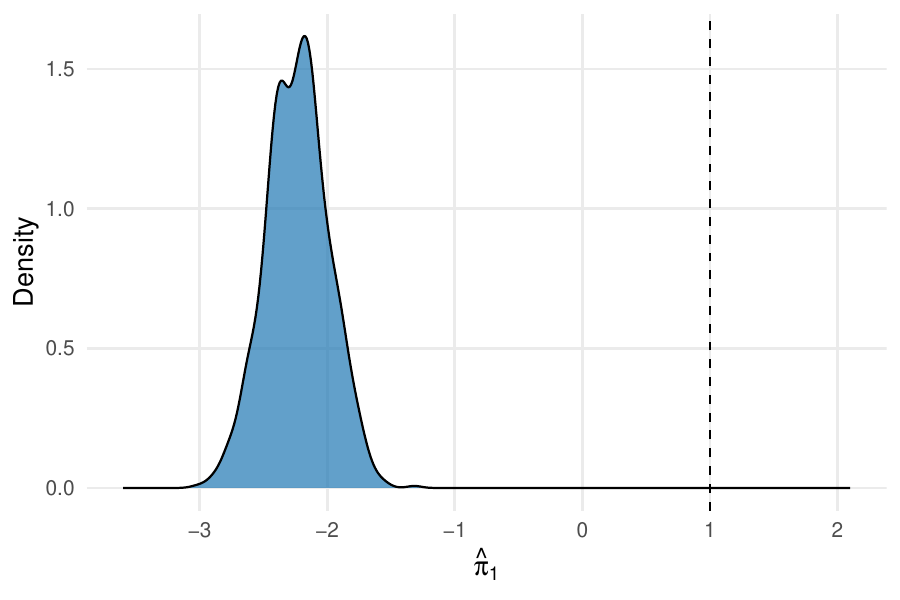}
    {\footnotesize (b) Direct effect}
  \end{minipage}
  \caption{Monte Carlo simulation of TWFE estimators}
  \label{fig:simulation_example}
  \vspace{0.3em}
  \par\footnotesize\raggedright
  \emph{Notes.} The figure reports the Monte Carlo distribution of the TWFE estimates based on 1000 simulation replications. Panel (a) shows the distribution of the TWFE coefficient on $tD_i$ from a regression of $Y_{it}$ on $tD_i$ with unit and time fixed effects, corresponding to the total effect. Panel (b) shows the distribution of the TWFE coefficient on $tD_i$ from a regression of $Y_{it}$ on $tD_i$ and $M_{it}$ with unit and time fixed effects, which is commonly interpreted as a direct effect. Vertical dashed lines indicate the true ATT in panel (a) and the true ADTT in panel (b).
\end{figure}

\subsection{Identification of direct effects}\label{sec:path-id}

{
The diagnostic results above show that neither the unconditional DiD contrast nor a TWFE regression controlling for the mediator generally recovers the direct effect. 
Hence, we propose an alternative approach; condition on the full observed mediator path \((M_{i0},M_{i1})\), construct a DiD contrast within each path, and then average these contrasts over the treated path distribution.
}

Define the conditional DiD estimand as
\begin{eqnarray*}
\tau^{DiD}(m_0,m_1)&\equiv& \mathbb{E}[\Delta Y_i\mid D_i=1,M_{i0}=m_0,M_{i1}=m_1]\\
&&-\ \mathbb{E}[\Delta Y_i\mid D_i=0,M_{i0}=m_0,M_{i1}=m_1].
\end{eqnarray*}

{
\begin{assumption}[Mediator path overlap and regularity]\label{ass:pathoverlap}
The treatment is nondegenerate conditional on the mediator path, $0<\mathbb P(D_i=1\mid M_{i0},M_{i1})<1$ a.s.; the relevant first moments are finite.
\end{assumption}
}

\begin{proposition}[Identification of direct treatment effect]\label{prop:direffect}
Consider the structural model \eqref{seq1}. Under Assumptions \ref{ass1p},
\ref{ass2p}, and \ref{ass:pathoverlap}, 
 the conditional average direct effect on the treated is identified as the conditional DiD estimand,
   $$ CADTT(m_0,m_1) \equiv\mathbb{E} [ \theta_i\mid D_i=1, M_{i0}=m_0,M_{i1}=m_1]=\tau^{DiD}(m_0,m_1)
    .$$

\vspace{0.3cm}
 Consequently, the average direct effect on the treated is identified,
    $$ ADTT  \equiv \mathbb{E} [ \theta_i\mid D_i=1]=\mathbb{E}\big[ \tau^{DiD}(M_{i0},M_{i1}) \big| D_i =1\big].$$
    \end{proposition}
\begin{proof}
    See Appendix \ref{pf_prop:direffect}.
\end{proof}

The proposition shows that the conditional average direct effect on the treated is identified by the DiD estimand conditional on the mediator path. What is surprising is that identification does not require the associated conditional parallel trends assumption to hold. To see this, let
$\Delta Y_i(0)\equiv Y_{i1}(0,M_{i1}(0))-Y_{i0}(0,M_{i0}(0))$. We show that 
\begin{align*}
&\mathbb E[\Delta Y_i(0)\mid D_i=1,M_{i1}=m_1,M_{i0}=m_0]
-\mathbb E[\Delta Y_i(0)\mid D_i=0,M_{i1}=m_1,M_{i0}=m_0]\\
&= \  -\mathbb E[\beta_{i1}]\cdot\mathbb E[\alpha_i\mid D_i=1,M_{i0}=m_0,M_{i1}=m_1],
\end{align*}
which is nonzero if both the average mediator effect on the outcome, $\mathbb E[\beta_{i1}]$,
and the conditional average treatment effect on the mediator, $\mathbb E[\alpha_i\mid D_i=1,M_{i0}=m_0,M_{i1}=m_1]$,
are different from zero. In general, when treatment affects the mediator and the mediator in turn affects the outcome, the standard conditional parallel trends assumption will fail. The reason is that the observed mediator path (referred to as \textit{factual path}) is contaminated by the treatment, and when the mediator matters to explain the variations in the outcome, this contamination may translate into different trends for treatment and control groups.  However, under Assumption \ref{ass1p} and \ref{ass2p}, parallel trends conditional on untreated potential mediator path (referred to as \textit{counterfactual path}) holds, that is,
\begin{align*}
&\mathbb E[Y_{i1}(0,m_1)-Y_{i0}(0,m_0)\mid D_i=1,M_{i1}(0)=m_1,M_{i0}(0)=m_0]\\
&\qquad =\mathbb E[Y_{i1}(0,m_1)-Y_{i0}(0,m_0)\mid D_i=0,M_{i1}(0)=m_1,M_{i0}(0)=m_0].
\end{align*}
The challenge is that the DiD estimand conditional on the counterfactual mediator path is infeasible. The result in Proposition \ref{prop:direffect} shows that the DiD conditional on the factual mediator path identifies the conditional average direct effect on the treated.

By integrating over the distribution of the factual mediator path conditional on the treatment group, we can identify the ADTT as
\begin{eqnarray*}
    \mathbb E[\theta_i\mid D_i=1] &=&\int \tau^{DiD}(m_0,m_1) dF_{M_{i0},M_{i1}\mid D_i=1}(m_0,m_1)\\
    &=&\int \mathbb E[\Delta Y_i \mid D_i=1, M_{i0}=m_0, M_{i1}=m_1]\\
    && \qquad \qquad - \mathbb E[\Delta Y_i \mid D_i=0, M_{i0}=m_0, M_{i1}=m_1] dF_{M_{i0},M_{i1}\mid D_i=1}(m_0,m_1).
\end{eqnarray*}

\begin{remark}[Binary and discrete mediator paths]
The conditional DiD argument extends to discrete or mixed mediators. 
These settings are special cases of the extension considered in Section \ref{sec:nonabsorbing-mediation}.
\end{remark}

\subsection{Identification of indirect and mediator effects} In general, while we show that the average direct effect is identifiable under our assumptions, the identification of the average indirect effect requires stronger assumptions, which may not be plausible in some applications. 
Define $\alpha^{DiD}\equiv\mathbb E[\Delta M_i\mid D_i=1]-\mathbb E[\Delta M_i\mid D_i=0],$
and consider the following assumption.
\begin{assumption}[Stable average mediator effect]\label{ass:stablebeta}
    $\mathbb E[\beta_{i1}]=\mathbb E[\beta_{i0}]$
\end{assumption}
This assumption rules out any time trend in the coefficient $\beta_{it}$ (e.g., $\beta_{it}=\beta_i+f(t)$). However, it allows for time heterogeneity in $\beta_{it}$. It only requires that $\beta_{it}$ has the same mean over time (e.g., $\beta_{it}=\beta_i+\tau_{it}$, where $\mathbb E[\tau_{it}]=0$).
When this assumption is combined with the previous unconditional parallel trends assumption, the standard DiD estimand identifies the average total effect on the treated. The average indirect effect on the treated is therefore obtained as the residual difference between the standard DiD estimand and the integral of the conditional DiD over the conditional distribution of the factual mediator path for the treated group. The following proposition holds.

\begin{proposition}\label{prop:medeffect}
    {Consider the structural model \eqref{seq1}.}
    Suppose $\alpha^{DiD}\neq0$. 
    Then under Assumptions \ref{ass1p}, \ref{ass2p}, \ref{ass:pathoverlap}, and \ref{ass:stablebeta}, $\mathbb E[\beta_{i1}]$ is identified as
    \begin{eqnarray*}
        \mathbb E[\beta_{i1}]=\frac{\tau^{DiD}-\int \tau^{DiD}(M_{i1}=m_1, M_{i0}=m_0) dF_{M_{i1},M_{i0}\mid D_i=1}(m_1,m_0)}{\alpha^{DiD}}.
    \end{eqnarray*}
\end{proposition}

\begin{proof} 
See Appendix \ref{pf_prop:medeffect}.
\end{proof}

\section{Potential outcome formulation}\label{sec:po-framework} 
{
We now extend the mediator path identification argument from the linear structural model to the general potential outcome framework. The structural model in Sections \ref{anaF}--\ref{sec:baseline-id} is useful for showing why conventional estimands fail and why path conditioning works. 
The potential outcome formulation below states the same logic directly in terms of potential outcomes and potential mediator paths.

Consider
}
\begin{eqnarray}\label{seqgeneral}
\left\{ \begin{array}{lcl}
     Y_{it} &=&  Y_{it}(1,M_{it}(1)) D_i+Y_{it}(0,M_{it}(0))(1-D_i),\\ \\
     M_{it} &=& M_{it}(1)D_i+M_{it}(0)(1-D_i),
     \end{array} \right.
\end{eqnarray}
where $Y_{it}(d,m_t)$ and $M_{it}(d)$ are potential outcome and potential mediator when the treatment $D_i$ and the mediator $M_{it}$ are exogenously set to $d$ and $m_t$, respectively. 
\begin{assumption}[Conditional parallel trends]\label{conditionalPT}
        \begin{eqnarray*}
    &&\mathbb E[Y_{i1}(0,m_1)-Y_{i0}(0,m_0)\mid D_i=1,M_{i1}(0)=m_1,M_{i0}(0)=m_0]\\
    &&=\mathbb E[Y_{i1}(0,m_1)-Y_{i0}(0,m_0)\mid D_i=0,M_{i1}(0)=m_1,M_{i0}(0)=m_0] \text{ for all } m_1, m_0.
    \end{eqnarray*}
        \end{assumption}
This assumption states that conditional on the counterfactual mediator path $(M_{i0}(0),M_{i1}(0))$, the average potential outcome trend is independent of treatment status. As discussed earlier, Model \eqref{seq1} and Assumptions \ref{ass1p}--\ref{ass2p} imply this assumption. This assumption allows for a violation of parallel trends conditional on the observed mediator path $(M_{i0},M_{i1})$. 

\begin{assumption}[Conditional orthogonality between potential mediator and untreated potential outcome trends]\label{conditionalorthorganity}
\begin{eqnarray*}
&& \mathbb E[Y_{i1}(0,m_1)-Y_{i0}(0,m_0) \mid D_i=d,M_{i1}(0)=m_1,M_{i0}(0)=m_0]\\
&& = \mathbb E[Y_{i1}(0,m_1)-Y_{i0}(0,m_0)\mid D_i=d, M_{i1}(1)=m_1,M_{i0}(1)=m_0] \text{ for all } m_1, m_0, d.
\end{eqnarray*}
       \end{assumption}
This assumption states that conditional on treatment status, the average potential outcome trend is independent of the potential mediator path. This assumption is also implied by Model \eqref{seq1} and Assumptions \ref{ass1p}--\ref{ass2p}. 

The framework considered here covers settings where the sequential ignorability assumption in \cite{imai2010identification} could be violated. We leverage the time dimension in the data structure to remove time-invariant factors that could compromise the validity of sequential ignorability.  Assumptions \ref{conditionalPT} and \ref{conditionalorthorganity} could be viewed as a generalization of sequential exogeneity to DiD settings. 
As is customary in the DiD literature, we also rely on the assumption that the treatment has no anticipatory effects on the outcome and mediator (Assumption \ref{noanticipation}). 
\begin{assumption}[No anticipation effects]\label{noanticipation}
$\,$
\begin{enumerate}[(i)]
    \item \label{noanticipation1} $Y_{i0}(1,m)=Y_{i0}(0,m),$ a.s. for all $m$.
    \item \label{noanticipation2} $M_{i0}(1)=M_{i0}(0),$ a.s.
\end{enumerate}
\end{assumption}
The following proposition holds.
\begin{proposition}\label{prop:generalresultADTT}
In the potential outcome model \eqref{seqgeneral},  under Assumptions \ref{ass:pathoverlap}, \ref{conditionalPT}, \ref{conditionalorthorganity},
and \ref{noanticipation}(\ref{noanticipation1}), the results in Proposition \ref{prop:direffect} hold. 

\end{proposition}
As we discussed in the structural model in the previous section, the unconditional parallel trends assumption may not hold because of the heterogeneity in the mediator effect on the outcome and the dependence between the mediator path and the treatment. However, if one is willing to assume that parallel trends holds, then one can recover the average indirect effect. 

\begin{assumption}[Unconditional parallel trends]\label{unconditionalPT}
    \begin{eqnarray*}
        \mathbb E[Y_1(0,M_1(0))-Y_0(0,M_0(0))\mid D=1]=\mathbb E[Y_1(0,M_1(0))-Y_0(0,M_0(0))\mid D=0].
    \end{eqnarray*}
\end{assumption}

Under Assumptions \ref{noanticipation} and \ref{unconditionalPT}, we can show that $\tau^{DiD}=ADTT+AITT$. Therefore, the AITT is identified as the difference between the DiD estimand and the identified ADTT. Hence, the following proposition holds.
\begin{proposition}\label{prop:generalresultAITT}
In the potential outcome model \eqref{seqgeneral}, under Assumptions \ref{ass:pathoverlap}, \ref{conditionalPT}, \ref{conditionalorthorganity}, \ref{noanticipation}, and \ref{unconditionalPT}, the AITT is identified: $$AITT=\tau^{DiD}-\int \tau^{DiD}(m_0,m_1) dF_{M_{i0},M_{i1}\mid D_i=1}(m_0,m_1).$$ 
\end{proposition}
The proof of this proposition is straightforward from the discussion above and is therefore omitted.

\section{Multiple treatment periods}\label{sec:multiple}

\subsection{Multiple treatment periods with no dynamic treatment effects}\label{sec:nonabsorbing-mediation}
Consider the following model specification: for $t\in \mathcal T\equiv \{0,1,\ldots,T\}$,
\begin{equation}\label{eq:multiplenonstaggered}
\left\{
\begin{array}{lcl}
Y_{it}&=&\theta_{it}D_{it}+\beta_{it}M_{it}+\lambda_i+\eta_t+U_{it},\\[0.2cm]
M_{it}&=&\psi_t(D_{it},\gamma_i,\tilde{V}_{it}),\\[0.2cm]
D_{it}&=&h_t(\lambda_i,\gamma_i,\varepsilon_{it}),
\end{array}
\right.
\end{equation}
where $D_{it}$ is a binary treatment indicator for individual $i$ at time $t$, and $\theta_{it}$ is the direct effect of the treatment on the outcome, which is allowed to vary across individuals and time. The variables $\tilde{V}_{it}$ and $\varepsilon_{it}$ could be multivariate. All other variables are the same as described in the model \eqref{seq1}. In the baseline period $(t=0)$, no individuals are treated ($D_{i0}=0$ for all $i$). This model assumes that past and future treatment statuses do not influence the current outcome and mediator. A similar restriction has been imposed on the outcome in \cite{dechaisemartin2020two} in the absence of mediators. 

The functions $\psi_t$ and $h_t$ are unrestricted. For example,  when the mediator $M$ is continuous, we can set $\psi_t(D_{it},\gamma_i,\tilde{V}_{it})=\alpha_{it} D_{it} + \gamma_i+\delta_t + V_{it}$, where $\tilde{V}_{it}=(\alpha_{it},V_{it})$. When the mediator $M$ is binary, we can set $\psi_t(D_{it},\gamma_i,\tilde{V}_{it})=\mathbbm{1}\{\alpha_{it} D_{it} + \gamma_i+\delta_t + V_{it}>0\}.$ A generalized version of the identifying assumption in the two-period setting is given in Assumption \ref{ass:multiplenonstaggered} below. 
\begin{assumption}\label{ass:multiplenonstaggered}$\,$
    \begin{enumerate}[(i)]
        \item  $(\lambda_i, \gamma_i,\{\varepsilon_{it}\}_{t\in \mathcal T})\ \indep\ (\{\tilde{V}_{it},U_{it},\beta_{it}\}_{t\in \mathcal T})$.
   \item  $\{\tilde{V}_{it}\}_{t\in \mathcal T}\ \indep\ (\{U_{it},\beta_{it}\}_{t\in \mathcal T})\mid (\lambda_i, \gamma_i,\{\varepsilon_{it}\}_{t\in \mathcal T})$.
   \item $0<\mathbb P(D_{it}=1\mid M_{i0},M_{it})<1$ a.s.
    \end{enumerate}
\end{assumption}
Under Assumption \ref{ass:multiplenonstaggered}, the standard parallel trends assumption may not hold for the outcome and the mediator. To see this, define $\Delta Y_{it}\equiv Y_{it}-Y_{i0}$, $\Delta M_{it} \equiv M_{it}-M_{i0}$, and $\Delta Y_{it}(0)\equiv Y_{it}(0,M_{it}(0))-Y_{i0}(0,M_{i0}(0))$. We have $\Delta Y_{it}(0)=\beta_{it}M_{it}(0)-\beta_{i0} M_{i0}(0)+\eta_t-\eta_0 + U_{it}-U_{i0}.$ Then, under Assumption \ref{ass:multiplenonstaggered}, 
\begin{eqnarray*}
    \mathbb E[\Delta Y_{it}(0)\mid D_{it}=1]-\mathbb E[\Delta Y_{it}(0)\mid D_{it}=0] &=& \mathbb E[\beta_{it}] (\mathbb E[M_{it}(0)\mid D_{it}=1]-\mathbb E[M_{it}(0)\mid D_{it}=0])\\
    &&-\mathbb E[\beta_{i0}] (\mathbb E[M_{i0}(0)\mid D_{it}=1]-\mathbb E[M_{i0}(0)\mid D_{it}=0]).
\end{eqnarray*}
From the above equality, it appears that if the mediator average effect is stable (i.e., $\mathbb E[\beta_{it}]=\mathbb E[\beta_{i0}]$) and parallel trends holds for the mediator (i.e., $\mathbb E[\Delta M_{it}(0)\mid D_{it}=1]-\mathbb E[\Delta M_{it}(0)\mid D_{it}=0]=0$), then parallel trends holds for the outcome. But in general, parallel trends may not hold for the outcome. It is important to also note that parallel trends conditional on the observed mediator path $(M_{it},M_{i0})$ may not hold either.\footnote{In fact, under Assumption \ref{ass:multiplenonstaggered}, $\mathbb E[\Delta Y_{it}(0)\mid D_{it}=1, M_{it}, M_{i0}]-\mathbb E[\Delta Y_{it}(0)\mid D_{it}=0, M_{it}, M_{i0}] = \mathbb E[\beta_{it}] (\mathbb E[M_{it}(0)\mid D_{it}=1, M_{it}, M_{i0}]-\mathbb E[M_{it}(0)\mid D_{it}=0, M_{it}, M_{i0}])$.} However, parallel trends conditional on the untreated potential mediator path $(M_{it}(0),M_{i0}(0))$ holds under Assumption \ref{ass:multiplenonstaggered}. 
The mediator can be continuous, discrete, or mixed. 

The proposition below generalizes the results in Proposition \ref{prop:direffect}. \begin{proposition}\label{prop:directeffect_multiple}
Under Assumption \ref{ass:multiplenonstaggered}, the conditional average direct effect on the contemporaneous treated group is identified as
\begin{eqnarray*}
    \mathbb E[\theta_{it}\mid D_{it}=1,M_{it}=m_t,M_{i0}=m_0] &=& \mathbb E[\Delta Y_{it} \mid D_{it}=1,M_{it}=m_t,M_{i0}=m_0]\\
    && \qquad -\mathbb E[\Delta Y_{it} \mid D_{it}=0,M_{it}=m_t,M_{i0}=m_0].
\end{eqnarray*}
The above estimand can be integrated over the distribution of the vector $(M_{it},M_{i0})$ conditional on the contemporaneous treatment group $D_{it}=1$ to obtain the average direct effect on the treated group. More precisely, the following holds.
\begin{eqnarray}
    \mathbb E[\theta_{it}\mid D_{it}=1] 
    &=&\int \mathbb E[\Delta Y_{it} \mid D_{it}=1, M_{it}=m_t, M_{i0}=m_0]\label{eq:multiADTT}\\
    && \qquad \qquad - \mathbb E[\Delta Y_{it} \mid D_{it}=0, M_{it}=m_t, M_{i0}=m_0] dF_{M_{it},M_{i0}\mid D_{it}=1}(m_t,m_0). \nonumber
\end{eqnarray}
\end{proposition}
\begin{proof}
    See Appendix \ref{pf_prop:direffectmultiple}.
\end{proof}
When $\psi_t(1,\gamma_i,\tilde{V}_{it})=\psi_t(0,\gamma_i,\tilde{V}_{it})$ a.s., then $M_{it}$ is a standard time-varying covariate. The identification strategy developed in this paper still applies.

To understand the relevance of the results in Proposition \ref{prop:directeffect_multiple}, we devote a special attention to the unconditional DiD estimand $\mathbb E[\Delta Y_{it} \mid D_{it}=1]-\mathbb E[\Delta Y_{it} \mid D_{it}=0]$. Under Assumption \ref{ass:multiplenonstaggered}, we can show that
\begin{eqnarray*}
    \mathbb E[\Delta Y_{it} \mid D_{it}=1]-\mathbb E[\Delta Y_{it} \mid D_{it}=0] &=& \mathbb E[\theta_{it} \mid D_{it}=1]\\
    && + \mathbb E[\beta_{it}](\mathbb E[M_{it}\mid D_{it}=1]-\mathbb E[M_{it}\mid D_{it}=0])\\
    && -\mathbb E[\beta_{i0}](\mathbb E[M_{i0}\mid D_{it}=1]-\mathbb E[M_{i0}\mid D_{it}=0]).
\end{eqnarray*}
The term $\mathbb E[\beta_{it}](\mathbb E[M_{it}\mid D_{it}=1]-\mathbb E[M_{it}\mid D_{it}=0]) -\mathbb E[\beta_{i0}](\mathbb E[M_{i0}\mid D_{it}=1]-\mathbb E[M_{i0}\mid D_{it}=0])$ is the bias that stems from ignoring the mediator $M$. We call it the \textit{mediator bias}. If the average mediator effect is stable over time,  (i.e., $\mathbb E[\beta_{it}]=\mathbb E[\beta_{i0}]$), then we can recover this effect as the ratio between the DiD estimand for $Y$ net of the identified average direct effect and the DiD estimand for $M$
$$
\mathbb E[\beta_{i0}] = \frac{(\mathbb E[\Delta Y_{it} \mid D_{it}=1]-\mathbb E[\Delta Y_{it} \mid D_{it}=0] )-\mathbb E[\theta_{it}\mid D_{it}=1]}{\mathbb E[\Delta M_{it} \mid D_{it}=1]-\mathbb E[\Delta M_{it} \mid D_{it}=0]},
$$
where $\mathbb E[\theta_{it}\mid D_{it}=1]$ is identified in \eqref{eq:multiADTT}, provided that the denominator is not zero.\footnote{
For example, if $M_{it} = \psi_t(D_{it},\gamma_i,\tilde{V}_{it})=\gamma_i + \delta_t + \tilde{V}_{it}$, then $M_{it}$ is a standard time-varying covariate that is not affected by $D_{it}$, and $\mathbb E[\Delta M_{it} \mid D_{it}=1]-\mathbb E[\Delta M_{it} \mid D_{it}=0] = 0$.} Note that we do not require parallel trends for the mediator to identify this quantity.

First, when the mediator has no effect on the outcome (i.e., $\beta_{it}=0$ for all $i$ and $t$), then the unconditional DiD estimand $\mathbb E[\Delta Y_{it} \mid D_{it}=1]-\mathbb E[\Delta Y_{it} \mid D_{it}=0]$ identifies the average direct effect on the treated, $\mathbb E[\theta_{it} \mid D_{it}=1]$. In such a case $\Delta Y_{it}=\theta_{it} D_{it} + \eta_t-\eta_0 + U_{it} -U_{i0}$, and we have
\begin{eqnarray*}
\mathbb E[\Delta Y_{it} \mid D_{it}] &=& \mathbb E[\theta_{it}\mid D_{it}] D_{it} + \eta_t-\eta_0+ \mathbb E[U_{it} - U_{i0} \mid D_{it}],\\
&=& \mathbb E[\theta_{it}\mid D_{it}] D_{it} + \eta_t-\eta_0+ \mathbb E[U_{it} - U_{i0}],\ \text{ under Assumption \ref{ass:multiplenonstaggered}},\\
&=& \theta_t D_{it} + \tilde{\eta}_t.
\end{eqnarray*}
One can run a cross-sectional regression for each period $t\in \{1,\ldots,T\},$
\begin{eqnarray}
\Delta Y_{it} &=& \theta_t D_{it} + \tilde{\eta}_t + e_{it}.
\end{eqnarray}
One can alternatively run the saturated regression to obtain the whole parameter vector $(\theta_1,\ldots,\theta_T)$ with a joint covariance matrix,
\begin{eqnarray}
\Delta Y_{it} = \sum_{s=1}^T \tilde{\eta}_s \mathbbm{1}\{t=s\}+ \sum_{s=1}^T \theta_s D_{it}\mathbbm{1}\{t=s\} + e_{it}.
\end{eqnarray}

Second, when $\psi_t(D_{it},\gamma_i,\tilde{V}_{it})=\alpha_{it} D_{it} + \gamma_i+\delta_t + V_{it}$, then the mediator bias simplifies to
$$\mathbb E[\beta_{it}]\,\mathbb E[\alpha_{it}]
+\mathbb E[\beta_{it}-\beta_{i0}]
\big(\mathbb E[\gamma_i\mid D_{it}=1]-\mathbb E[\gamma_i\mid D_{it}=0]\big).$$
This mediator bias becomes the average indirect effect when the mediator effect $\beta_{it}$ is on average stable over time (i.e., $\mathbb E[\beta_{it}]=\mathbb E[\beta_{i0}]$) or the treatment selection variable, $D_{it}$, is independent of the mediator fixed effects. In such a case, the standard parallel trends assumption holds for the outcome and the mediator (because of the linear two-way fixed effects specification). 

\begin{remark}[Multiple pre-treatment periods]\label{rem:pretreatment}
    When there are multiple pre-treatment periods, we can use the fact that the identification of $\mathbb E[\theta_{it}\mid D_{it}=1]$ does not depend on a specific pre-treatment period to check for violations of our identifying assumptions. More precisely, suppose that there are $T_0$ pre-treatment periods. For $t_0\in \{-T_0,\ldots, 0\}$, define $\Delta Y_{it}^{t_0}\equiv Y_{it}-Y_{it_0}$. Since no individuals were treated in the pre-treatment periods, we have $D_{it_0}=0$ for all $t_0\in \{-T_0,\ldots, 0\}$. Suppose that model~\eqref{eq:multiplenonstaggered} and Assumption \ref{ass:multiplenonstaggered} hold for $t\in \mathcal T=\{-T_0, \ldots,0,\ldots,T\}.$ Then, 
    \begin{eqnarray*}
    \mathbb E[\theta_{it}\mid D_{it}=1] 
    &=&\int \mathbb E[\Delta Y_{it}^{t_0} \mid D_{it}=1, M_{it}=m_t, M_{it_0}=m_{t_0}]\\
    && \qquad \qquad - \mathbb E[\Delta Y_{it}^{t_0} \mid D_{it}=0, M_{it}=m_t, M_{it_0}=m_{t_0}] dF_{M_{it},M_{it_0}\mid D_{it}=1}(m_t,m_{t_0}). \nonumber
\end{eqnarray*}
for all  $t_0\in \{-T_0,\ldots, 0\}$. So, we can test for whether the above estimand is constant across values of $t_0$ for each $t\in \{1,\ldots,T\}$. Note, however, that the non-rejection of this implication does not imply that our model assumptions hold. This is the analog of pre-trends test in our setting.  

We illustrate this implication in Example \ref{ex:testimp} below.
\end{remark}

\begin{example} \label{ex:testimp}
We consider two cases where one satisfies the identifying assumptions (Case A) and the other does not because of a differential trend in the outcome unobservable (Case B).

Units $i=1,\ldots,n=4{,}000$ are observed in periods
$t\in\{-T_0,\ldots,0,1\}$ with $T_0=5$ pre-treatment periods before the baseline period $t=0$ and a single post-treatment period, $T=1$. 
The data follow
\begin{align}
Y_{it} &= \theta_i D_{it} + \beta_{it} M_{it} + \lambda_i + \eta_t + U_{it},
\label{eq:sim-y}\\
M_{it} &= \alpha D_{it} + \gamma_i + \delta_t + V_{it},
\label{eq:sim-m}\\
D_i &= \mathbbm{1}\{a_0 + a_\lambda \lambda_i + a_\gamma \gamma_i
       + \varepsilon_i > 0\},
\label{eq:sim-d}
\end{align}
with $\alpha=0.5$, $\eta_t=0.10\,t$, $\delta_t=0.15\,t$, $a_0=-0.25$ and $a_\lambda=a_\gamma=0.60$ where $D_{it}=D_i\,\mathbbm{1}\{t=1\}$. 
The unit fixed effects $(\lambda_i,\gamma_i)$ are bivariate standard normal with correlation $0.5$, and $\varepsilon_i\sim\mathcal{N}(0,1)$. 
The direct effect is $\theta_i = 0.5 + 0.3\,\lambda_i + 0.3\,\xi_i$ with $\xi_i\sim\mathcal{N}(0,1)$, so $\theta_{it}=\theta_i$ for all $t$. 
The return to the mediator exhibits trends over time,
\begin{equation}
\beta_{it} = \bar\beta + \beta_g\,t + b_i + \omega_{it},
\qquad \bar\beta=1,\ \beta_g=0.2.
\label{eq:sim-beta}
\end{equation}

More details on the DGP are provided in Appendix \ref{apx_ex1}. 
Case A uses a design where Assumption~\ref{ass:multiplenonstaggered} holds. Case B replaces $U_{it}$ by
$U_{it}+\tau\,t\,D_i$ with $\tau=0.15$, a linear trend in the outcome unobservable of the treated units, which violates Assumption~\ref{ass:multiplenonstaggered}(i), because $U_{it}$ then depends on $D_i$ and hence on $(\lambda_i,\gamma_i,\varepsilon_i)$. Under this violation, the estimand that uses base period $t_0$ equals $\mathbb E[\theta_{i1}\mid D_{i1}=1]+\tau(1-t_0)$, so the bias is $0.90$ at $t_0=-5$ and $0.15$ at $t_0=0$.

Let \begin{eqnarray*}
     \tau^{DiD}_{t_0} &\equiv& \int \mathbb E[\Delta Y_{i1}^{t_0} \mid D_{i1}=1, M_{i1}=m_1, M_{it_0}=m_{t_0}]\\
    && \qquad \qquad - \mathbb E[\Delta Y_{i1}^{t_0} \mid D_{i1}=0, M_{i1}=m_1, M_{it_0}=m_{t_0}] dF_{M_{i1},M_{it_0}\mid D_{i1}=1}(m_1,m_{t_0}). \nonumber
\end{eqnarray*}
For the sake of exposition, we first plot the confidence intervals for the difference $\tau^{DiD}_{t_0}-\tau^{DiD}_0$ for $t_0 \in \{-5,\ldots,0\}$. If $\tau^{DiD}_{t_0}$ is constant for all $t_0 \in \{-5,\ldots,0\}$, then these confidence intervals contain zero. Second, on the same graph, we plot the confidence interval for the estimand $\tau^{DiD}_0$ and the true ADTT. Under our assumptions, $ADTT=\tau^{DiD}_0$, and the confidence interval for $\tau^{DiD}_0$ should contain the true ADTT. 

\begin{figure}[!htbp]
\centering
\includegraphics[width=\linewidth]{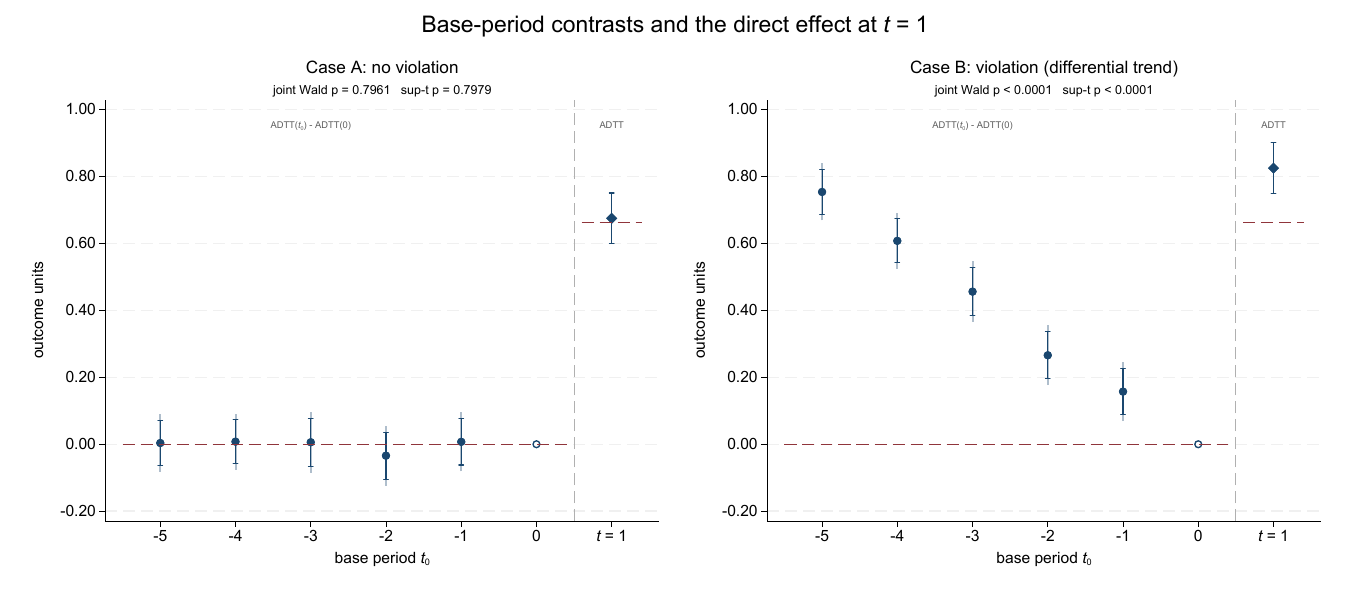}
\caption{Base-period contrasts and the direct effect at $t=1$, simulated
data}
\label{fig:pretest-sim}
\begin{minipage}{\linewidth}
\footnotesize
\emph{Notes.} Left of the dashed vertical divider, each panel plots the contrasts $\hat\tau^{DiD}_{t_0}-\hat\tau^{DiD}_{0}$ for $t_0=-5,\ldots,-1$, where $\hat\tau^{DiD}_{t_0}$ is the estimate of $\tau^{DiD}_{t_0}$ that uses base period $t_0$, on one simulated sample of $n=4{,}000$ units.
The hollow marker at $t_0=0$ marks the reference contrast, which is zero by construction. 
Dark caps are pointwise 95\% intervals, and 
light spikes are the simultaneous 95\% band \citep{montiel2019simultaneous}.
Right of the divider, the diamond is $\hat\tau^{DiD}_{0}$, the ADTT at
$t=1$ estimated from base period $t_0=0$, with its pointwise 95\% interval.
The contrasts and the level are both in outcome units, so the two regions share the vertical axis. 
The dashed horizontal reference is zero left of the divider and the realized target (the average of $\theta_i$ over the treated units, $0.662$) right of it.
The joint Wald $p$ and sup-$t$ $p$-values in the subtitles test $H_0\colon\tau^{DiD}_{-5}=\cdots=\tau^{DiD}_{0}$. 
The sup-$t$ test rejects at the 5\% level exactly when a light
spike misses the zero line, so the test can be read off the figure.
\end{minipage}
\end{figure}

Figure~\ref{fig:pretest-sim} summarizes the results.
In case A, the realized target $0.662$ lies within all six pointwise confidence intervals, the profile is flat, and neither test rejects: joint Wald $p=0.796$, sup-$t$ $p=0.798$ (see Appendix \ref{sec:estimation} for a description of these test statistics). 
In case B, the six ADTT estimates show trends from $1.58$ at $t_0=-5$ to $0.82$ at $t_0=0$ (the differences are from $0.75$ at $t_0=-5$ to $0.18$ at $t_0=-1$), and both $p$-values are below $0.0001$.
Note also that the closest ADTT confidence interval at $t_0=0$ does not cover the realized target $0.662$.

\end{example}

 \subsection{Multiple treatment periods with staggered adoption and dynamic treatment effects}\label{sec:staggered-adoption}
We modify the model specification \eqref{eq:multiplenonstaggered} to allow the outcome to depend on past or future treatments. The staggered adoption model we consider is the following. For $t\in \mathcal T=\{0,\ldots,T\}$,
\begin{eqnarray}\label{eq:multiplestaggered}
\left\{ \begin{array}{lcl}
     Y_{it} &=&  (\sum_{g=1}^T\theta_{it}^g D_i^g)\mathbbm{1}\{t>0\} +\beta_{it} M_{it} + \lambda_i +\eta_t +U_{it},\\ \\
     M_{it} &=& (\sum_{g=1}^T\alpha_{it}^g D_i^g)\mathbbm{1}\{t>0\}+\gamma_i + \delta_t + V_{it},\\ \\
     D_{it} &=& \mathbbm{1}\{h(\lambda_i,\gamma_i, \varepsilon_i)\leq t,t>0\},
     \end{array} \right.
\end{eqnarray}
where $D_i^g$ is the indicator for whether individual $i$ receives treatment for the first time in period $g$ or not, and $\theta_{it}^g$ is the direct effect of individual $i$ receiving treatment for the first time in period $g$ on her outcome in period $t$.
Let $\mathcal G \equiv \{1, \ldots, T, \infty\}$.
\begin{assumption}\label{ass:multiplestaggered}$\,$
    \begin{enumerate}[(i)]
        \item  $(\lambda_i,\gamma_i,\varepsilon_i)\ \indep\ (\{V_{it},U_{it},\beta_{it},\{\alpha_{it}^g: g\in \mathcal T\}\}_{t\in \mathcal T})$.
   \item  $(\{V_{it},\{\alpha_{it}^g: g\in \mathcal T\}\}_{t\in \mathcal T})\ \indep\ (\{U_{it},\beta_{it}\}_{t\in \mathcal T}) \mid \lambda_i, \gamma_i,\varepsilon_i$.
   \item $0<\mathbb P(D_{i}^g=1\mid M_{i0},M_{it})<1$ a.s. for all $g \in \mathcal G$.
    \end{enumerate}
\end{assumption}
Assumption \ref{ass:multiplestaggered} implies  $(\{V_{it},\{\alpha_{it}^g: g\in \mathcal T\}\}_{t\in \mathcal T},\lambda_i, \gamma_i,\varepsilon_i)\ \indep\ (\{U_{it}, \beta_{it}\}_{t\in \mathcal T})$.
We can show that under this assumption, the conditional average direct treatment effect on the treated is identified:
\begin{eqnarray*}
    \mathbb E[\theta_{it}^g\mid D_i^g=1, M_{it}=m_t, M_{i0}=m_0]&=&\mathbb E[\Delta Y_{it} \mid D_i^g=1, M_{it}=m_t, M_{i0}=m_0]\\
    &&- \mathbb E[\Delta Y_{it} \mid D_i^{\infty}=1, M_{it}=m_t, M_{i0}=m_0],
\end{eqnarray*}
where $D_i^{\infty}$ is an indicator for the never-treated group. We can obtain the average direct effect on the treated as
\begin{eqnarray*}
    \mathbb E[\theta_{it}^g\mid D_{i}^g=1] 
    &=&\int \mathbb E[\Delta Y_{it} \mid D_{i}^g=1, M_{it}=m_t, M_{i0}=m_0]\\
    && \qquad \qquad - \mathbb E[\Delta Y_{it} \mid D_{i}^{\infty}=1, M_{it}=m_t, M_{i0}=m_0] dF_{M_{it},M_{i0}\mid D_{i}^g=1}(m_t,m_0). \nonumber
\end{eqnarray*}

\section{Feedback: including lagged outcomes}\label{sec:laggedoutcome}
\subsection{Simple DiD with lagged outcomes}
We first consider the following simple DiD model: for $t\in \mathcal T\equiv \{-2,-1,0,1\}$,
\begin{equation}\label{eq:simpleDiDlag}
\left\{
\begin{array}{lcl}
Y_{it}&=&\theta_{it}D_{it}+\rho Y_{it-1}+\lambda_i+\eta_t+U_{it},\\[0.2cm]
D_{it}&=&h_t(\lambda_i,\varepsilon_{it}) \text{ for } t>0,\\[0.2cm]
D_{it}&=&0 \text{ for } t\leq 0,
\end{array}
\right.
\end{equation}
where $(\lambda_i,\varepsilon_{i1})\ \indep\ (U_{i(-1)},U_{i0},U_{i1}).$ We have 
\begin{eqnarray*}
    Y_{i0}-Y_{i(-1)}=\rho (Y_{i(-1)}-Y_{i(-2)})+\eta_0-\eta_{(-1)}+U_{i0}-U_{i(-1)}.
\end{eqnarray*}
On the one hand, since $(\lambda_i,\varepsilon_{i1})\ \indep\ (U_{i(-1)},U_{i0},U_{i1})$, we have $D_{i1}\ \indep\ U_{i0}-U_{i(-1)}$ (exogeneity). On the other hand, the model implies $Y_{i(-1)}-Y_{i(-2)}=(\rho-1)Y_{i(-2)}+\lambda_i+\eta_{(-1)}+U_{i(-1)},$ which likely implies $Cov(Y_{i(-1)}-Y_{i(-2)},D_{i1})\neq 0$ (relevance). In addition, $D_{i1}$ satisfies the exclusion restriction. Therefore, $D_{i1}$ is a valid IV for $Y_{i(-1)}-Y_{i(-2)}$, provided that the relevance condition holds. Hence, $\rho$ is identified as
$$\rho=\frac{Cov(Y_{i0}-Y_{i(-1)},D_{i1})}{Cov(Y_{i(-1)}-Y_{i(-2)},D_{i1})}.$$
Once we identify $\rho$, we can rewrite the outcome model as
$$\Dot{Y}_{it}=\theta_{it}D_{it}+\lambda_i+\eta_t+U_{it},$$
where $\Dot{Y}_{it}\equiv Y_{it}-\rho Y_{it-1}.$ The parallel trends assumption holds in this transformed outcome model and we identify the ATT as 
\begin{eqnarray}
    \mathbb E[\theta_{i1}\vert D_{i1}=1]=\mathbb E[\Delta \Dot{Y}_{i1}\vert D_{i1}=1]-\mathbb E[\Delta \Dot{Y}_{i1}\vert D_{i1}=0]\equiv \theta_1^{\text{AR-DiD}},
\end{eqnarray}
where $\Delta \Dot{Y}_{i1}=\Dot{Y}_{i1}-\Dot{Y}_{i0}.$ We call this new estimand \textit{autoregressive difference-in-differences} (AR-DiD). 

\subsection{DiD with mediators and lagged outcomes}
We now add mediators to the DiD model with lagged outcomes. We consider the following model specification: for $t\in \mathcal T\equiv \{-2,-1,0,1\}$,
\begin{equation}\label{eq:multiplenonstaggeredlag}
\left\{
\begin{array}{lcl}
Y_{it}&=&\theta_{it}D_{it}+\rho Y_{it-1}+\beta_{it}M_{it}+\lambda_i+\eta_t+U_{it},\\[0.2cm]
M_{it}&=&\psi_t(D_{it},\gamma_i,\tilde{V}_{it}),\\[0.2cm]
D_{it}&=&h_t(\lambda_i,\gamma_i,\varepsilon_{it}),
\end{array}
\right.
\end{equation}
where all variables are the same as defined in the previous section. The only addition is the inclusion of the lag outcome $Y_{it-1}$ in the model \eqref{eq:multiplenonstaggered} and we have three pre-treatment periods so that $D_{it}=0$ for $t\in\{-2,-1,0\}$. 

Model \eqref{eq:multiplenonstaggeredlag} implies
\begin{eqnarray*}
    Y_{i0}-Y_{i(-1)} = \rho (Y_{i(-1)}-Y_{i(-2)}) + \beta_{i0}M_{i0}-\beta_{i(-1)} M_{i(-1)}+\eta_{0}-\eta_{(-1)}+U_{i0}-U_{i(-1)}. 
\end{eqnarray*}
Under Assumption \ref{ass:multiplenonstaggered}, $D_{i1}\ \indep\ (\beta_{i0}M_{i0}-\beta_{i(-1)} M_{i(-1)}+\eta_{0}-\eta_{(-1)}+U_{i0}-U_{i(-1)}) \vert (M_{i0},M_{i(-1)})$.\footnote{Note that under Assumption \ref{ass:multiplenonstaggered}, $(D_{i1},M_{i0},M_{i(-1)})\ \indep\ (U_{i0}, U_{i(-1)},\beta_{i0}, \beta_{i(-1)})$.} Therefore, $D_{i1}$ could serve as an IV for $Y_{i(-1)}-Y_{i(-2)}$, and we can write:
\begin{eqnarray*}
    Cov(Y_{i0}-Y_{i(-1)},D_{i1}\vert M_{i0},M_{i(-1)}) = \rho Cov(Y_{i(-1)}-Y_{i(-2)},D_{i1}\vert M_{i0},M_{i(-1)}).
\end{eqnarray*}
Hence, we can identify $\rho$ as
\begin{eqnarray*}
 \rho = \frac{\mathbb E[Cov(Y_{i0}-Y_{i(-1)},D_{i1}\vert M_{i0},M_{i(-1)})] }{\mathbb E[Cov(Y_{i(-1)}-Y_{i(-2)},D_{i1}\vert M_{i0},M_{i(-1)})]},
\end{eqnarray*}
provided $\mathbb E[Cov(Y_{i(-1)}-Y_{i(-2)},D_{i1}\vert M_{i0},M_{i(-1)})]\neq 0$. Now that we identify $\rho$, we define $\Dot{Y}_{it}\equiv Y_{it}-\rho Y_{it-1}$, and rewrite model \eqref{eq:multiplenonstaggeredlag} in the framework of model~\eqref{eq:multiplenonstaggered} as
\begin{equation*}
\left\{
\begin{array}{lcl}
\Dot{Y}_{it}&=&\theta_{it}D_{it}+\beta_{it}M_{it}+\lambda_i+\eta_t+U_{it},\\[0.2cm]
M_{it}&=&\psi_t(D_{it},\gamma_i,\tilde{V}_{it}),\\[0.2cm]
D_{it}&=&h_t(\lambda_i,\gamma_i,\varepsilon_{it}).
\end{array}
\right.
\end{equation*}
The identification result in Proposition \ref{prop:directeffect_multiple} applies to the transformed outcome $\Dot{Y}_{it}$, and we have
\begin{eqnarray}
                \mathbb E[\theta_{it}\mid D_{it}=1] 
    &=&\int \mathbb E[\Delta \Dot{Y}_{it} \mid D_{it}=1, M_{it}=m_t, M_{i0}=m_0]\label{eq:multiADTTlag}\\
    && \qquad \qquad - \mathbb E[\Delta \Dot{Y}_{it} \mid D_{it}=0, M_{it}=m_t, M_{i0}=m_0] dF_{M_{it},M_{i0}\mid D_{it}=1}(m_t,m_0), \nonumber
\end{eqnarray}
where $\Delta \Dot{Y}_{it} \equiv \Dot{Y}_{it} - \Dot{Y}_{i0}$.

\section{Extensions}\label{sec:extensions}
{
\subsection{Multivariate mediators}\label{sec:multi-mediators}

In practice, there are multiple mediators that are responses to the treatment and influence the main outcome of interest. This section extends our identification results to multivariate mediators.  

\subsubsection{Model and identification of average direct effect}

Let $\boldsymbol M_{it}\in\mathbb R^K$ be a vector of continuous mediators. The
multivariate model is
\begin{equation}\label{eq:mmulti}
\left\{
\begin{array}{lcl}
Y_{it} &=& \theta_i tD_i+\boldsymbol\beta_{it}'\boldsymbol M_{it}
          +\lambda_i+\eta_t+U_{it},\\[0.2cm]
\boldsymbol M_{it} &=& \boldsymbol\alpha_i tD_i
          +\boldsymbol\gamma_i+\boldsymbol\delta_t+\boldsymbol V_{it},\\[0.2cm]
D_i &=& h(\lambda_i,\boldsymbol\gamma_i,\varepsilon_i).
\end{array}
\right.
\end{equation}
where
$\boldsymbol\beta_{it},\boldsymbol\alpha_i,\boldsymbol\gamma_i,
\boldsymbol\delta_t,\boldsymbol V_{it}\in\mathbb R^K$. The average indirect
effect on the treated is
\[
AITT
=
\mathbb E\!\left[
Y_{i1}(0,\boldsymbol M_{i1}(1))
-
Y_{i1}(0,\boldsymbol M_{i1}(0))
\mid D_i=1
\right]
=
\mathbb E[\boldsymbol\beta_{i1}'\boldsymbol\alpha_i\mid D_i=1].
\]
Assumption \ref{ass:multi_exog} below generalizes Assumptions \ref{ass1p}, \ref{ass2p}, \ref{ass:pathoverlap} to a multivariate mediator setting.  
\begin{assumption}\label{ass:multi_exog}$\,$
\begin{enumerate}[(i)]
    \item $(\lambda_i,\boldsymbol\gamma_i,\varepsilon_i)
\ \indep\
(\{\boldsymbol V_{it},U_{it},\boldsymbol\beta_{it}\}_{t=0,1},
\boldsymbol\alpha_i)$.
\item $(\{\boldsymbol V_{it}\}_{t=0,1},\boldsymbol\alpha_i)
\ \indep\
(\{U_{it},\boldsymbol\beta_{it}\}_{t=0,1})
\mid
(\lambda_i,\boldsymbol\gamma_i,\varepsilon_i).$
\item $0<\mathbb P(D_i=1 \mid \boldsymbol M_{i0}, \boldsymbol M_{i1})<1$ a.s.
\end{enumerate}
\end{assumption}
Let $\Delta Y_i\equiv Y_{i1}-Y_{i0}$, $\Delta\boldsymbol M_i\equiv \boldsymbol M_{i1}-\boldsymbol M_{i0}.$
For
$(\boldsymbol m_0,\boldsymbol m_1)$ in the treated mediator-path support,
define the path-conditional DiD estimand
\begin{eqnarray*}
\tau^{DiD}(\boldsymbol m_0,\boldsymbol m_1)
&\equiv&
\mathbb E[\Delta Y_i\mid D_i=1,\boldsymbol M_{i0}=\boldsymbol m_0,\boldsymbol M_{i1}=\boldsymbol m_1]\\
&& -\ \mathbb E[\Delta Y_i\mid D_i=0,\boldsymbol M_{i0}=\boldsymbol m_0,\boldsymbol M_{i1}=\boldsymbol m_1].
\end{eqnarray*}
Define
the vector mediator DiD estimand
\[
\boldsymbol\alpha^{DiD}
\equiv
\mathbb E[\Delta\boldsymbol M_i\mid D_i=1]
-
\mathbb E[\Delta\boldsymbol M_i\mid D_i=0].
\]

\begin{proposition}[Direct effect and mediator response with multivariate mediators]
\label{prop:multi_direct_alpha}
Consider the multivariate-mediator structural model \eqref{eq:mmulti}. Under Assumption \ref{ass:multi_exog},
$\mathbb E[\boldsymbol\alpha_i]=\boldsymbol\alpha^{DiD},$
and
\begin{align*}
ADTT\equiv \mathbb E[\theta_i\mid D_i=1]&=\int
\tau^{DiD}(\boldsymbol m_0,\boldsymbol m_1)
\,dF_{\boldsymbol M_{i0},\boldsymbol M_{i1}\mid D_i=1}
(\boldsymbol m_0,\boldsymbol m_1) \\ &=\mathbb E[\tau^{DiD}(\boldsymbol M_{i0},\boldsymbol M_{i1}) \mid D_i=1].
\end{align*}
\end{proposition}
The proof of this proposition is similar to that of Proposition \ref{prop:direffect} and is therefore omitted.

Under the stable mediator effect assumption ($\mathbb E[\boldsymbol \beta_{i1}]=\mathbb E[\boldsymbol \beta_{i0}]$), we can recover the average indirect effect AITT as the residual difference between the standard DiD estimand and the average direct effect ADTT. In general, when the mediator effect is heterogeneous over time, we cannot recover the AITT without further assumptions.

\subsection{Extension to nonlinear difference-in-differences}
We consider the model below where the treatment is binary, $D\in \{0,1\}$, and there are two time periods, $t\in \{0,1\}$. This model is an extension of the changes-in-changes model in \cite{athey2006identification} to include a mediator. We drop the subscript $i$ here for simplicity.

\begin{eqnarray}\label{seqnonlineardid}
\left\{ \begin{array}{lcl}
     Y_{t} &=&  g_t(t D,M_{t},U_{t})\\ \\
     M_{t} &=& \psi_t(t D, V_{t})
     \end{array} \right.
\end{eqnarray}
In the above model, $U_t$ is a scalar unobservable, while $V_t$ is a vector of unobservables. $g_t$ and $\psi_t$ are unknown functions. 
\begin{assumption}\label{ass:nonlineardid} $\,$
    \begin{enumerate}[(i)]
        \item \label{ass1:nonlineardid} $U_1 \mid D=d, M_0=m_0, M_1=m_1 \sim^{d} U_0 \mid D=d, M_0=m_0, M_1=m_1$.
        \item \label{ass2:nonlineardid} The function $g_t(0,m_t,\cdot)$ is continuous and strictly increasing.
        \item \label{ass3:nonlineardid} $\operatorname{supp}\{U_t \mid D=1, M_0=m_0, M_1=m_1\} \subseteq \operatorname{supp}\{U_t \mid D=0, M_0=m_0, M_1=m_1\}$.
    \end{enumerate}
\end{assumption}
Assumption \ref{ass:nonlineardid}(\ref{ass1:nonlineardid}) states that the distribution of the outcome unobserved heterogeneity $U_t$ conditional on $(D,M_0,M_1)$ is stable over time. Assumption \ref{ass:nonlineardid}(\ref{ass2:nonlineardid}) is a smoothness and monotonicity condition that ensures that the function $g_t(0,m_t,\cdot)$ is invertible. Assumption \ref{ass:nonlineardid}(\ref{ass3:nonlineardid}) is a support condition that requires that the support of the outcome unobserved heterogeneity $U_t$ for the treated group conditional on the joint vector $(M_0,M_1)$ is included in that of the control group. These assumptions combined imply that the counterfactual distribution $\mathbb P(Y_1(0,m_1) \leq y \mid D=1, M_0=m_0, M_1=m_1)$ is identified, which helps identify the average direct effect ADTT. 

Define $Y_t(0,m_t)\equiv g_t(0,m_t,U_t)$, and $Y_t(0,M_t(0))\equiv g_t(0,M_t(0),U_t)$. 
\begin{proposition}\label{prop:nonlinear}
    Under Assumptions \ref{ass:pathoverlap} and \ref{ass:nonlineardid}, we have
    \begin{eqnarray*}
        && \mathbb P(Y_1(1,m_1) \leq y \mid D=1, M_0=m_0, M_1=m_1) = \mathbb P(Y_1 \leq y \mid D=1, M_0=m_0, M_1=m_1),\\
       && \mathbb P(Y_1(0,m_1) \leq y \mid D=1, M_0=m_0, M_1=m_1) =  \\
       && \qquad \qquad F_{Y_0 \mid D=1, M_0=m_0, M_1=m_1}\left(F^{-1}_{Y_0 \mid D=0, M_0=m_0, M_1=m_1}(F_{Y_1 \mid D=0, M_0=m_0, M_1=m_1}(y))\right),
    \end{eqnarray*}
    for $F_{M_0,M_1\mid D=1}$-almost every $(m_0,m_1)$.
\end{proposition}

\begin{proof}
    See Appendix \ref{apx:nonlinear}.
\end{proof}

{
}

\section{Empirical illustration: railroads and market access}\label{sec:empirical}

The mediator-path identification result writes the $ADTT$ as a conditional DiD estimand averaged over the mediator-path distribution of treated units. We implement this estimand using the doubly robust estimator of \cite{sant2020doubly}. To keep the application aligned with the notation used throughout the paper, $Y$ denotes the outcome, $M$ the mediator, $D$ the treatment, and $\boldsymbol X$ the baseline covariates throughout this section. We reserve $\boldsymbol W$ for the target-specific conditioning vector used in a given doubly robust specification.

\subsection{Questions, sample, and empirical design}

Donaldson and Hornbeck's principal object is the aggregate contribution of the national railroad network to the agricultural sector in 1890 \citep{donaldson2016railroads}. They first estimate how changes in county market access between 1870 and 1890 are capitalized into agricultural land values. They then calculate how much each county's market access would fall if all railroads were removed from the 1890 transportation network and aggregate the implied land-value losses. Their headline conclusion is that removing all railroads would reduce the total value of U.S.\ agricultural land by $60.2\%$, with annual losses equal to $3.22\%$ of GNP.

In their framework, market access is constructed to summarize the direct and indirect spatial general-equilibrium effects of changes in the national railroad network. Their use of the terms ``direct'' and ``indirect'' therefore does not correspond to the natural direct and indirect treatment effects studied in this paper. Likewise, the local-railroad controls in their Table II assess whether the estimated relationship between market access and land value is robust to conditioning on own and nearby railroad construction; those controls are not used to construct a causal mediation decomposition.

Our illustration instead asks a local-adoption question. Let $t=0$ denote 1870 and $t=1$ denote 1890. For county $o$, let $Y_{ot}$ denote the log value of agricultural land and buildings, let $M_{ot}$ denote log cost-based market access, and define the baseline geographic covariates as
\[
\boldsymbol X_o
=
(\mathrm{lon}_o,\mathrm{lat}_o)',
\]
where $\mathrm{lon}_o$ and $\mathrm{lat}_o$ denote, respectively, the longitude and latitude of the geographic centroid of county $o$.
The treatment indicator $D_o$ equals one if county $o$ had no railroad track in 1870 but had track in 1890, and equals zero if the county had no track in either year. Under the maintained assumptions, we report the total effect of local railroad adoption on $Y$, the mediator response of $M$, the $ADTT$ identified through mediator-path conditioning, and the residual $AITT$.

We begin with the $2{,}327$-county balanced regression sample used in the Donaldson--Hornbeck baseline specification. Restricting the sample to counties without railroad track in 1870 excludes the $1{,}272$ counties that already had track, leaving $1{,}055$ counties observed in both periods. Of these counties, $730$ are treated and $325$ are controls. Because $M_{ot}$ is a network object, it can change for control counties when railroad track is constructed elsewhere. By contrast, $D_o$ represents whether the county itself adopts railroad track between the two periods and stays zero for all counties in the baseline period.
We report standard errors using state-clustered CR0 stacked-sandwich covariances and use the delta method where needed.\footnote{The moment system and the covariance construction are given in Appendix \ref{sec:estimation}.} The results without baseline covariates are discussed in Table \ref{tab:estimates_dh_nocov} in the Appendix~\ref{apx:dh-nocov}.

Throughout the application, the average direct effect on the treated, $ADTT$, is the component not transmitted through the mediator, the measured cost-based market access. It can therefore contain depot-town activity, local commerce, amenities, and omitted dimensions or spillovers of market access. Because each county's outcome and market access evolve with the realized national network, the estimated $ADTT$ is a local-adoption effect in that network, not the no-interference effect of changing the county's track status while holding the rest of the network fixed. 

\subsection{TWFE and doubly robust specifications}

For the original-control TWFE benchmark, define the separate cubic geographic basis
\[
\boldsymbol b_{\mathrm{geo}}(\boldsymbol X_o)
=
\big(
\mathrm{lon}_o,\mathrm{lon}_o^2,\mathrm{lon}_o^3,
\mathrm{lat}_o,\mathrm{lat}_o^2,\mathrm{lat}_o^3
\big)'.
\]
For county $o$ in state $\mathfrak s(o)$ and period $t\in\{0,1\}$, we estimate
\begin{equation}\label{eq:twfe-dh}
Y_{ot}
=
\pi_1 tD_o
+\pi_2 M_{ot}
+\phi_o
+\zeta_{\mathfrak s(o),t}
+\boldsymbol b_{\mathrm{geo}}(\boldsymbol X_o)'\boldsymbol\chi_t
+\xi_{ot},
\end{equation}
where $\phi_o$, $\zeta_{\mathfrak s(o),t}$, and $\boldsymbol b_{\mathrm{geo}}(\boldsymbol X_o)'\boldsymbol\chi_t$ represent county fixed effects, state-by-period fixed effects, and the cubic geographic trends used in Equation (13) of \cite{donaldson2016railroads}, and $tD_o$ denotes local railroad adoption.

In Table \ref{tab:estimates_dh}, the TWFE total effects are obtained by omitting $M_{ot}$ from \eqref{eq:twfe-dh}, and the direct effect and the $M\to Y$ effect are estimates of $\pi_1$ and $\pi_2$, respectively. The TWFE indirect effects are estimated as the total coefficient minus $\pi_1$. Recall that the direct and indirect effects obtained from the TWFE regressions are descriptive, not causal, by Claim \ref{claim:interpretTWFE}.

The two TWFE benchmarks use the same $1{,}055$ counties and equal observation weights, but they differ in their covariate selection. The original-control column contains state-by-period effects and the cubic geographic basis in \eqref{eq:twfe-dh}, whereas the same-controls column replaces them with common period effects and linear trends in $\boldsymbol X_o$, matching the sample, weights, and baseline covariates in the DR specifications. Comparing that column with DR therefore holds the observed covariate set fixed and helps separate differences due to the estimator from differences due to the richer original-control specification.

The doubly robust columns implement the improved estimator of \cite{sant2020doubly}. 
For the total effect and the mediator response, the conditioning vector is
\[
\boldsymbol W_o=\boldsymbol X_o.
\]
For the direct effect, the linear mediator-path specification uses
\[
\boldsymbol W_o
=
(M_{o0},M_{o1},\boldsymbol X_o')',
\]
whereas the additive-quadratic\footnote{``Quadratic'' in Table \ref{tab:estimates_dh} refers to separate squared terms in the two mediator coordinates and does not include the interaction $M_{o0}M_{o1}$.} mediator-path specification uses
\[
\boldsymbol W_o
=
(M_{o0},M_{o0}^2,M_{o1},M_{o1}^2,\boldsymbol X_o')'.
\]

The DR specification cannot retain the original state and cubic controls. Eight of the $37$ represented states contain treated counties but no controls, so treated-support overlap fails within those states; one additional state contains controls only. We therefore define the same-controls TWFE column with the linear geographic pair $\boldsymbol X_o$, which is the covariate set used in the DR fits.
}

\input{DH_output/tab_estimates_dh}

\subsection{Average effects}

As discussed above, Table \ref{tab:estimates_dh} shows contrasts in direct and indirect effects between TWFE approaches and DR specifications under our assumptions, while the total effect estimates are stable across designs when the covariate set is held fixed.
In particular, the direct effect estimate is statistically indistinguishable from zero under the TWFE specification with linear geographic controls but is $0.429$ $(0.177)$ under our preferred linear path DR specification.
Also, we notice the sign reversal in indirect effect estimates from $0.258$ $(0.118)$ to $-0.039$ $(0.141)$, yet the latter is not statistically significant.
Under the additive-quadratic path model, the DR estimates remain qualitatively similar to the preferred linear path DR estimates, but the direct estimate loses precision.
The overlap diagnostics suggest treating this as a sensitivity analysis, not as a second preferred specification.

\section{Summary and discussion}\label{Conclusion}

We study DiD designs where treatment changes a time-varying mediator and the mediator also affects the outcome. In such settings, policymakers may want to know the direct effect of the policy on an outcome and how much of the policy effect operates through the mediator. However, a challenge is that the mediator is endogenous and affected by the treatment, and the conventional DiD estimand may not have a clear causal interpretation. Hence, neither ignoring the mediator nor adding it as a time-varying control in a TWFE regression generally helps identify the desired policy effects.

Instead, Proposition \ref{prop:direffect} identifies $CADTT(m_0,m_1)$ by comparing outcome changes across treatment groups at a common observed mediator path, and integrating this contrast over treated paths yields $ADTT$. The result does not require parallel trends conditional on the factual mediator path under our framework, and Claim \ref{claim:interpretTWFE} explains why including $M_{it}$ in TWFE does not produce the same object. Under the addition of Assumption \ref{ass:stablebeta}, the unconditional outcome and mediator DiD moments also identify the indirect effect and the average mediator coefficient when the mediator response is nonzero.
One lesson for applied work, not only practical but also important, is that a post-treatment mediator should not be treated as an ordinary control in a DiD regression. Conditioning on the mediator path can be useful, but the target parameter and the assumptions differ from those underlying standard covariate-adjusted DiD.

The empirical illustration reinforces this point. In the railroad application, the linear DR path model estimates $ADTT=0.429$ and $AITT=-0.039$, with the latter imprecisely estimated. On the other hand, the same-controls TWFE yields the direct effect of $0.090$ that cannot be statistically distinguished from zero. The second application, Alcohol Prohibition and Farm Productivity, which is reported in the Appendix, similarly shows the difference between the standard TWFE framework and our approach in the models with covariates.

We extend our identification approach to include past outcomes. However, there are several limitations pointing to useful directions for future work. First, the identifying assumptions considered here restrict the dependence between mediator shocks and outcome shocks after conditioning on the unobservables that determine treatment selection. Also, the identification of the indirect effect relies on the stability of the average mediator effect, and relaxing this restriction would require additional sources of variation or partial identification arguments.

\bibliographystyle{jpe}
\bibliography{mybib}

\clearpage
\appendix

\section{Proofs of main results}

\subsection{Proof of Lemma \ref{lem:uncond_indep}}\label{pf_lem:uncond_indep}

\begin{proof}
Assumptions \ref{ass1p} and \ref{ass2p} imply
\begin{eqnarray*}
\left\{ \begin{array}{lcl}
     (U_{i0}, U_{i1}, \beta_{i0}, \beta_{i1})\ \indep\ (V_{i0}, V_{i1}, \alpha_i),\\ \\
    (U_{i0}, U_{i1}, \beta_{i0}, \beta_{i1})\ \indep\ (\lambda_i,\gamma_i,\varepsilon_i),\\ \\ 
    (\lambda_i, \gamma_i,\varepsilon_i)\ \indep\ (V_{i0}, V_{i1}, \alpha_i),
     \end{array} \right.
\end{eqnarray*}
Therefore, $(U_{i0}, U_{i1}, \beta_{i0}, \beta_{i1})\ \indep\ (V_{i0}, V_{i1}, \alpha_i, \lambda_i, \gamma_i,\varepsilon_i)$. In fact,
\begin{eqnarray*}
        F_{(V_{i0},V_{i1}, \alpha_i,\varepsilon_i,\lambda_i, \gamma_i,U_{i0},U_{i1},\beta_{i0},\beta_{i1})} &=& F_{(V_{i0},V_{i1}, \alpha_i,U_{i0},U_{i1},\beta_{i0},\beta_{i1}) \mid (\lambda_i,\gamma_i,\varepsilon_i)} F_{(\lambda_i,\gamma_i,\varepsilon_i)},\\
        &=& F_{(V_{i0},V_{i1}, \alpha_i) \mid (\lambda_i,\gamma_i,\varepsilon_i)} F_{(U_{i0},U_{i1},\beta_{i0},\beta_{i1}) \mid (\lambda_i,\gamma_i,\varepsilon_i)} F_{(\lambda_i,\gamma_i,\varepsilon_i)}\ \\ &&\text{ under Assumption \ref{ass2p}},\\
        &=& F_{(V_{i0},V_{i1}, \alpha_i) \mid (\lambda_i,\gamma_i,\varepsilon_i)} F_{(U_{i0},U_{i1},\beta_{i0},\beta_{i1})} F_{(\lambda_i,\gamma_i,\varepsilon_i)}\ \\ &&\text{ under Assumption \ref{ass1p}},\\
       &=& F_{(V_{i0},V_{i1}, \alpha_i, \lambda_i,\gamma_i,\varepsilon_i)} F_{(U_{i0},U_{i1},\beta_{i0},\beta_{i1})}. 
    \end{eqnarray*}
\end{proof}

\subsection{Proof of Claim \ref{claim:interpretTWFE}} \label{pf_claim:interpretTWFE}

\begin{proof}

    Define 
    \begin{align*}
      \widetilde{D}_{i}	\equiv & \ D_{i}-\Pi(D_{i}\mid\Delta M_{i})
	= D_{i}-\mathbb{E}(D_{i})-\frac{\operatorname{Cov}(D_{i},\Delta M_{i})}{\operatorname{Var}(\Delta M_{i})}\left(\Delta M_i-\mathbb{E}(\Delta M_{i})\right) \\
    \widetilde{\Delta M_i}	\equiv & \ \Delta M_i-\Pi(\Delta M_i\mid D_i)
	= \Delta M_i-\mathbb{E}(\Delta M_i)-\frac{\operatorname{Cov}(D_{i},\Delta M_{i})}{\operatorname{Var}(D_i)}\left(D_i-\mathbb{E}(D_i)\right)
   \end{align*} 
   Since we only have two periods, the fixed effect estimand and the first difference estimand are numerically identical. Therefore, by the Frisch–Waugh–Lovell theorem,

   \begin{align*}
       \pi_{1}=&\frac{\operatorname{Cov}(\widetilde{D}_{i},\Delta Y_{i})}{\operatorname{Var}(\widetilde{D}_{i})}\\
       = & \frac{1}{\operatorname{Var}(\widetilde{D}_{i})}\Bigg\{ \operatorname{Cov}(D_{i},\theta_{i}D_{i})+\operatorname{Cov}(D_{i},\beta_{i1}\Delta M_{i})+\operatorname{Cov}\left(D_{i},(\beta_{i1}-\beta_{i0})M_{i0}\right)\\
       &  - \frac{\operatorname{Cov}(D_{i},\Delta M_{i})}{\operatorname{Var}(\Delta M_{i})}\left[\operatorname{Cov}(\Delta M_{i},\theta_{i}D_{i})+\operatorname{Cov}(\Delta M_{i},\beta_{i1}\Delta M_{i})+\operatorname{Cov}\left(\Delta M_{i},(\beta_{i1}-\beta_{i0})M_{i0}\right)\right]\Bigg\}\\
       = & \frac{1}{\operatorname{Var}(\widetilde{D}_{i})}\Bigg\{ \operatorname{Cov}(D_{i},\theta_{i}D_{i})+\mathbb{E}(\alpha_{i})\mathbb{E}(\beta_{i1})\operatorname{Var}(D_{i})+\mathbb{E}(\beta_{i1}-\beta_{i0})\operatorname{Cov}(D_{i},\gamma_{i})\\
       & -\frac{\mathbb{E}(\alpha_{i})\operatorname{Var}(D_{i})}{\operatorname{Var}(\Delta M_{i})}\left[\operatorname{Cov}(\Delta M_{i},\theta_{i}D_{i})+\mathbb{E}(\beta_{i1})\operatorname{Var}(\Delta M_{i})+\mathbb{E}(\beta_{i1}-\beta_{i0})\operatorname{Cov}(\Delta M_{i},M_{i0})\right]\Bigg\}\\
       =& \frac{1}{\operatorname{Var}(\widetilde{D}_{i})}\Bigg\{ \operatorname{Cov}(D_{i},\theta_{i}D_{i})+\mathbb{E}(\beta_{i1}-\beta_{i0})\operatorname{Cov}(D_{i},\gamma_{i})\\
       & -\frac{\mathbb{E}(\alpha_{i})\operatorname{Var}(D_{i})}{\operatorname{Var}(\Delta M_{i})}\left[\operatorname{Cov}(\Delta M_{i},\theta_{i}D_{i})+\mathbb{E}(\beta_{i1}-\beta_{i0})\operatorname{Cov}(\Delta M_{i},M_{i0})\right]\\
       & -\frac{\mathbb{E}(\alpha_{i})\operatorname{Var}(D_{i})}{\operatorname{Var}(\Delta M_{i})}\left[\mathbb{E}(\alpha_{i})\operatorname{Cov}(D_{i},\theta_{i}D_{i})-\mathbb{E}(\alpha_{i})\operatorname{Cov}(D_{i},\theta_{i}D_{i})\right]\Bigg\}\\
       = & \frac{1}{\operatorname{Var}(\widetilde{D}_{i})}\Bigg\{ \operatorname{Cov}(D_{i},\theta_{i}D_{i})\left(1-\frac{\mathbb{E}(\alpha_{i})^{2}\operatorname{Var}(D_{i})}{\operatorname{Var}(\Delta M_{i})}\right)\\
       &+\mathbb{E}(\beta_{i1}-\beta_{i0})\left(\operatorname{Cov}(D_{i},\gamma_{i})-\frac{\mathbb{E}(\alpha_{i})\operatorname{Var}(D_{i})}{\operatorname{Var}(\Delta M_{i})}\operatorname{Cov}(\Delta M_{i},M_{i0})\right)\\
       &-\frac{\mathbb{E}(\alpha_{i})\operatorname{Var}(D_{i})}{\operatorname{Var}(\Delta M_{i})}\left[\operatorname{Cov}(\Delta M_{i},\theta_{i}D_{i})-\mathbb{E}(\alpha_{i})\operatorname{Cov}(D_{i},\theta_{i}D_{i})\right]\Bigg\}\\
       = &\frac{\operatorname{Cov}(D_{i},\theta_{i}D_{i})}{\operatorname{Var}(D_i)}+\mathbb{E}(\beta_{i1}-\beta_{i0})\left(\frac{\operatorname{Cov}(D_{i},\gamma_{i})}{\operatorname{Var}(\widetilde{D}_{i})}-\frac{\mathbb{E}(\alpha_{i})\operatorname{Cov}(\Delta M_{i},M_{i0})}{\operatorname{Var}(\widetilde{\Delta M_{i}})}\right)\\
       &+\frac{\mathbb{E}(\alpha_{i})}{\operatorname{Var}(\widetilde{\Delta M_{i}})}\left(\mathbb{E}(\alpha_{i})\operatorname{Cov}(D_{i},\theta_{i}D_{i})-\operatorname{Cov}(\Delta M_{i},\theta_{i}D_{i})\right)\\
       = & \mathbb{E}(\theta_{i}\mid D_{i}=1)+\mathbb{E}(\beta_{i1}-\beta_{i0})\left(\frac{\operatorname{Cov}(D_{i},\gamma_{i})}{\operatorname{Var}(\widetilde{D}_{i})}-\frac{\mathbb{E}(\alpha_{i})\operatorname{Cov}(\Delta M_{i},M_{i0})}{\operatorname{Var}(\widetilde{\Delta M_{i}})}\right)\\
       &+\frac{\mathbb{E}(\alpha_{i})}{\operatorname{Var}(\widetilde{\Delta M_{i}})}\left(\mathbb{E}(\alpha_{i})\operatorname{Cov}(D_{i},\theta_{i}D_{i})-\operatorname{Cov}(\Delta M_{i},\theta_{i}D_{i})\right)
   \end{align*}

The second equality follows from Model~\ref{seq1} and Lemma \ref{lem:uncond_indep}, and the third equality follows from the following observations:

$\operatorname{Cov}\left(D_{i},\beta_{i1}\Delta M_{i}\right)=\mathbb{E}(\alpha_{i})\mathbb{E}(\beta_{i1})\operatorname{Var}(D_{i})$,

$\operatorname{Cov}\left(D_{i},(\beta_{i1}-\beta_{i0})M_{i0}\right)=\mathbb{E}(\beta_{i1}-\beta_{i0})\operatorname{Cov}(D_{i},\gamma_{i})$,

$\operatorname{Cov}(D_{i},\Delta M_{i})=\mathbb{E}(\alpha_{i})\operatorname{Var}(D_{i})$,

$\operatorname{Cov}\left(\Delta M_{i},\beta_{i1}\Delta M_{i}\right)=\mathbb{E}(\beta_{i1})\operatorname{Var}(\Delta M_{i})$, and

$\operatorname{Cov}\left(\Delta M_{i},(\beta_{i1}-\beta_{i0})M_{i0}\right)=\mathbb{E}(\beta_{i1}-\beta_{i0})\operatorname{Cov}(\Delta M_{i},M_{i0})$.

For the fourth equality, the term $\mathbb{E}(\beta_{i1})\operatorname{Var}(\Delta M_{i})$ cancels out, and we add and subtract $\mathbb{E}(\alpha_{i})\operatorname{Cov}(D_{i},\theta_{i}D_{i})$. The fifth equality then follows by reorganizing and factoring terms. The sixth equality holds by $\operatorname{Var}(\tilde{D}_{i})=\operatorname{Var}(D_{i})-\frac{\mathbb{E}(\alpha_{i})^{2}\operatorname{Var}(D_{i})^{2}}{\operatorname{Var}(\Delta M_{i})}$ and $\frac{\operatorname{Var}(D_{i})}{\operatorname{Var}(\tilde{D}_{i})}=\frac{\operatorname{Var}(\Delta M_{i})}{\operatorname{Var}(\widetilde{\Delta M_{i}})}$. The final equality follows from $\frac{\operatorname{Cov}(D_i,\theta_i D_i)}{\operatorname{Var}(D_i)}=\mathbb{E}(\theta_i \mid D_i=1)$

Analogously, 
\begin{align*}
    \pi_2 = & \frac{\operatorname{Cov}(\widetilde{\Delta M_{i}},\Delta Y_{i})}{\operatorname{Var}(\widetilde{\Delta M_{i}})} \\
    = & \frac{1}{\operatorname{Var}(\widetilde{\Delta M_{i}})}\Bigg\{ \operatorname{Cov}(\Delta M_{i},\theta_{i}D_{i})+\operatorname{Cov}(\Delta M_{i},(\beta_{i1}-\beta_{i0})M_{i0})+\operatorname{Cov}(\Delta M_{i},\beta_{i1}\Delta M_{i})\\
    & -\frac{\operatorname{Cov}(\Delta M_{i},D_{i})}{\operatorname{Var}(D_{i})}\left[\operatorname{Cov}(D_{i},\theta_{i}D_{i})+\operatorname{Cov}(D_{i},(\beta_{i1}-\beta_{i0})M_{i0})+\operatorname{Cov}(D_{i},\beta_{i1}\Delta M_{i})\right]\Bigg\}\\
    = & \frac{1}{\operatorname{Var}(\widetilde{\Delta M_{i}})}\Bigg\{ \operatorname{Cov}(\Delta M_{i},\theta_{i}D_{i})+\mathbb{E}(\beta_{i1}-\beta_{i0})\operatorname{Cov}(\Delta M_{i},M_{i0})+\mathbb{E}(\beta_{i1})\operatorname{Var}(\Delta M_{i})\\
    & -\mathbb{E}(\alpha_{i})\left[\operatorname{Cov}(D_{i},\theta_{i}D_{i})+\mathbb{E}(\beta_{i1}-\beta_{i0})\operatorname{Cov}(D_{i},\gamma_{i})+\mathbb{E}(\alpha_{i})\mathbb{E}(\beta_{i1})\operatorname{Var}(D_{i})\right]\Bigg\} \\
    = & \frac{1}{\operatorname{Var}(\widetilde{\Delta M_{i}})}\Bigg\{\mathbb{E}(\beta_{i1})\left[\operatorname{Var}(\Delta M_{i})-\mathbb{E}(\alpha_{i})^{2}\operatorname{Var}(D_i)\right]+\operatorname{Cov}(\Delta M_{i},\theta_{i}D_{i})-\mathbb{E}(\alpha_{i})\operatorname{Cov}(D_{i},\theta_{i}D_{i})\\
    & + \mathbb{E}(\beta_{i1}-\beta_{i0})\left[\operatorname{Cov}(\Delta M_{i},M_{i0})-\mathbb{E}(\alpha_{i})\operatorname{Cov}(D_{i},\gamma_{i})\right]\Bigg\}\\
    = & \mathbb{E}(\beta_{i1})+\frac{\mathbb{E}(\beta_{i1}-\beta_{i0})}{\operatorname{Var}(\widetilde{\Delta M_{i}})}\left[\operatorname{Cov}(\Delta M_{i},M_{i0})-\mathbb{E}(\alpha_{i})\operatorname{Cov}(D_{i},\gamma_{i})\right]\\
    & +\frac{1}{\operatorname{Var}(\widetilde{\Delta M_{i}})}\big[\operatorname{Cov}(\Delta M_{i},\theta_{i}D_{i})-\mathbb{E}(\alpha_{i})\operatorname{Cov}(D_{i},\theta_{i}D_{i})\big]
\end{align*}

    The third equality uses $\operatorname{Cov}(\Delta M_i,D_i)=\mathbb{E}(\alpha_i)\operatorname{Var}(D_i)$. The fourth equality is obtained by collecting terms, and the fifth equality follows from $\operatorname{Var}(\widetilde{\Delta M_{i}})=\operatorname{Var}(\Delta M_{i})-\mathbb{E}(\alpha_{i})^{2}\operatorname{Var}(D_{i})$.
\end{proof}

\subsection{Proof of Proposition \ref{prop:direffect}}\label{pf_prop:direffect}
\begin{proof}
    \begin{eqnarray*}
    && \tau^{DiD}(m_0,m_1)\\ &\equiv& \mathbb E[\Delta Y_i \mid D_i=1, M_{i1}=m_1, M_{i0}=m_0]- \mathbb E[\Delta Y_i \mid D_i=0, M_{i1}=m_1, M_{i0}=m_0],\\
    &=& \mathbb E[\theta_i D_i +\beta_{i1}M_{i1}-\beta_{i0}M_{i0}+\eta_1-\eta_0 + U_{i1}-U_{i0}\mid D_i=1, M_{i1}=m_1, M_{i0}=m_0]\\
    && - \mathbb E[\theta_i D_i +\beta_{i1}M_{i1}-\beta_{i0}M_{i0}+\eta_1-\eta_0 + U_{i1}-U_{i0}\mid D_i=0, M_{i1}=m_1, M_{i0}=m_0],\\
    &=& \mathbb E[\theta_i \mid D_i=1, M_{i1}=m_1, M_{i0}=m_0]\\
    &&+\mathbb E[\beta_{i1}m_{1}-\beta_{i0}m_{0}+\eta_1-\eta_0 + U_{i1}-U_{i0}\mid h(\lambda_i, \gamma_i,\varepsilon_i)=1, M_{i1}(1)=m_1, M_{i0}(1)=m_0]\\
    && - \mathbb E[\beta_{i1}m_{1}-\beta_{i0}m_{0}+\eta_1-\eta_0 + U_{i1}-U_{i0}\mid h(\lambda_i,\gamma_i,\varepsilon_i)=0, M_{i1}(0)=m_1, M_{i0}(0)=m_0],\\
    &=&\mathbb E[\theta_i \mid D_i=1, M_{i1}=m_1, M_{i0}=m_0]\\
    && + \mathbb E[\beta_{i1}m_{1}-\beta_{i0}m_{0}+\eta_1-\eta_0 + U_{i1}-U_{i0}\mid h(\lambda_i, \gamma_i,\varepsilon_i)=1,\\
    && \hspace{5.2cm} \alpha_i+\gamma_i+\delta_1+V_{i1}=m_1, \gamma_i+\delta_0+V_{i0}=m_0]\\
    && -\mathbb E[\beta_{i1}m_{1}-\beta_{i0}m_{0}+\eta_1-\eta_0 + U_{i1}-U_{i0}\mid h(\lambda_i,\gamma_i,\varepsilon_i)=0,\\
    && \hspace{5.2cm} \gamma_i+\delta_1+V_{i1}=m_1, \gamma_i+\delta_0+V_{i0}=m_0],\\
    &=&\mathbb E[\theta_i \mid D_i=1, M_{i1}=m_1, M_{i0}=m_0] + \mathbb E[\beta_{i1}m_{1}-\beta_{i0}m_{0}+\eta_1-\eta_0 + U_{i1}-U_{i0}]\\
    && -\mathbb E[\beta_{i1}m_{1}-\beta_{i0}m_{0}+\eta_1-\eta_0 + U_{i1}-U_{i0}]\\
    && \hspace{2cm}\text{ because } (U_{i0}, U_{i1}, \beta_{i0}, \beta_{i1})\ \indep\ (V_{i0}, V_{i1}, \alpha_i, \lambda_i, \gamma_i,\varepsilon_i),\\
    &=&\mathbb E[\theta_i \mid D_i=1, M_{i1}=m_1, M_{i0}=m_0].
\end{eqnarray*}
\end{proof}

\subsection{Proof of Proposition \ref{prop:medeffect}}
\label{pf_prop:medeffect}
\begin{proof}
    We have shown in the main text that 
    \begin{eqnarray*}
     \tau^{DiD} &=& \mathbb E[\theta_i \mid D_i=1] +\mathbb E[\beta_{i1}] \mathbb E[\alpha_i]\\
     && + \mathbb E[\gamma_i(\beta_{i1}-\beta_{i0})\mid h(\lambda_i,\gamma_i,\varepsilon_i)=1]-\mathbb E[\gamma_i(\beta_{i1}-\beta_{i0})\mid h(\lambda_i,\gamma_i,\varepsilon_i)=0].   
    \end{eqnarray*}
    Using the law of iterated expectations, we can write for $d\in\{0,1\}$
    \begin{eqnarray*}
        \mathbb E[\gamma_i(\beta_{i1}-\beta_{i0})\mid h(\lambda_i,\gamma_i,\varepsilon_i)=d] &=& \mathbb E[\gamma_i\mathbb E[(\beta_{i1}-\beta_{i0})\mid h(\lambda_i,\gamma_i,\varepsilon_i)=d,\gamma_i]\mid h(\lambda_i,\gamma_i,\varepsilon_i)=d],\\
        &=& \mathbb E[\gamma_i\mathbb E[\beta_{i1}-\beta_{i0}]\mid h(\lambda_i,\gamma_i,\varepsilon_i)=d]\\
        && \qquad \qquad \qquad  \text{ under Assumption \ref{ass1p}},\\
        &=& \mathbb E[\beta_{i1}-\beta_{i0}]\mathbb E[\gamma_i\mid h(\lambda_i,\gamma_i,\varepsilon_i)=d],\\
        &=& (\mathbb E[\beta_{i1}]-\mathbb E[\beta_{i0}])\mathbb E[\gamma_i\mid h(\lambda_i,\gamma_i,\varepsilon_i)=d],\\
        &=& 0 \ \text{ under Assumption \ref{ass:stablebeta}}
    \end{eqnarray*}
    Therefore, $\tau^{DiD} = \mathbb E[\theta_i \mid D_i=1] +\mathbb E[\beta_{i1}] \mathbb E[\alpha_i]$. Under Assumption \ref{ass1p}, $\mathbb E[\alpha_i]$ is identified as the DiD estimand for the mediator, $\mathbb E[\alpha_i]=\alpha^{DiD}$ because parallel trends for the mediator holds under our model. Hence,
    \begin{eqnarray*}
        \mathbb E[\beta_{i1}]&=& \frac{\tau^{DiD}-\mathbb E[\theta_i\mid D_i=1]}{\alpha^{DiD}},\\
        &=&\frac{\tau^{DiD}-\int \tau^{DiD}(M_{i1}=m_1, M_{i0}=m_0) dF_{M_{i1},M_{i0}\mid D_i=1}(m_1,m_0)}{\alpha^{DiD}}.
    \end{eqnarray*}
\end{proof}

\subsection{Proof of Proposition \ref{prop:generalresultADTT}}\label{apx:generalresult}
\begin{proof}
    We have 
    \begin{eqnarray*}
        \tau^{DiD}(m_0,m_1) &=& \mathbb E[Y_1-Y_0 \mid D=1,M_1=m_1,M_0=m_0]\\
        && -\ \mathbb E[Y_1-Y_0 \mid D=0,M_1=m_1,M_0=m_0],\\
        &=& \mathbb E[Y_1(1,m_1)-Y_0(1,m_0) \mid D=1,M_1=m_1,M_0=m_0]\\
        && \qquad \qquad -\mathbb E[Y_1(0,m_1)-Y_0(0,m_0) \mid D=0,M_1=m_1,M_0=m_0],\\
        &=& \mathbb E[Y_1(1,m_1)-Y_1(0,m_1) \mid D=1,M_1=m_1,M_0=m_0]\\
        && \qquad \qquad + \mathbb E[Y_1(0,m_1)-Y_0(1,m_0) \mid D=1,M_1=m_1,M_0=m_0]\\
        && \qquad \qquad -\mathbb E[Y_1(0,m_1)-Y_0(0,m_0) \mid D=0,M_1=m_1,M_0=m_0],\\
        &=& \mathbb E[Y_1(1,m_1)-Y_1(0,m_1) \mid D=1,M_1=m_1,M_0=m_0]\\
        && \qquad \qquad + \mathbb E[Y_1(0,m_1)-Y_0(0,m_0) \mid D=1,M_1(1)=m_1,M_0(1)=m_0]\\
        && \hspace{7.6cm} \text{ under Assumption~\ref{noanticipation}(\ref{noanticipation1})}\\
        && \qquad \qquad -\mathbb E[Y_1(0,m_1)-Y_0(0,m_0) \mid D=0,M_1(0)=m_1,M_0(0)=m_0],\\
        &=& \mathbb E[Y_1(1,M_1(1))-Y_1(0,M_1(1)) \mid D=1,M_1=m_1,M_0=m_0]\\
        && \qquad \qquad + \mathbb E[Y_1(0,m_1)-Y_0(0,m_0) \mid D=1,M_1(1)=m_1,M_0(1)=m_0] \\
        && \qquad \qquad -\mathbb E[Y_1(0,m_1)-Y_0(0,m_0) \mid D=1,M_1(0)=m_1,M_0(0)=m_0]\\
        && \hspace{7.6cm} \text{ under Assumption~\ref{conditionalPT}},\\
        &=& \mathbb E[Y_1(1,M_1(1))-Y_1(0,M_1(1)) \mid D=1,M_1=m_1,M_0=m_0]\\
        && \qquad \qquad + \mathbb E[Y_1(0,m_1)-Y_0(0,m_0) \mid D=1,M_1(0)=m_1,M_0(0)=m_0]\\
        && \hspace{7.6cm} \text{ under Assumption~\ref{conditionalorthorganity}}\\
        && \qquad \qquad -\mathbb E[Y_1(0,m_1)-Y_0(0,m_0) \mid D=1,M_1(0)=m_1,M_0(0)=m_0],\\
        &=& \mathbb E[Y_1(1,M_1(1))-Y_1(0,M_1(1)) \mid D=1,M_1=m_1,M_0=m_0]
    \end{eqnarray*}
{
Hence, by the law of iterated expectations,
\begin{eqnarray*}
    &&\int \tau^{DiD}(m_0,m_1)\, dF_{(M_1, M_0) \mid D=1}(m_1, m_0) \\
    &&\qquad = \mathbb E\bigl[\mathbb E[Y_1(1,M_1(1))-Y_1(0,M_1(1)) \mid D=1,M_1,M_0] \mid D=1\bigr]\\
    &&\qquad = \mathbb E[Y_1(1,M_1(1))-Y_1(0,M_1(1)) \mid D=1]\\
    &&\qquad = ADTT.
\end{eqnarray*}
}
\end{proof}

\subsection{Proof of Proposition \ref{prop:nonlinear}}\label{apx:nonlinear}

\begin{proof} 
All statements below hold for
$F_{M_0,M_1\mid D=1}$-almost every $(m_0,m_1)$.
The overlap condition in Assumption~\ref{ass:pathoverlap}
ensures that the observed conditional distributions
$F_{Y_t\mid D=d,M_0=m_0,M_1=m_1}$, $t,d\in\{0,1\}$,
are identified at these paths.
We have
\begin{eqnarray*}
    && \mathbb P(Y_1(0,m_1) \leq y \mid D=d, M_0=m_0, M_1=m_1)\\
    && = \mathbb P(g_1(0,m_1,U_1) \leq y \mid D=d, M_0=m_0, M_1=m_1),\\
    &&= \mathbb P(U_1 \leq g_1^{-1}(0,m_1;y) \mid D=d, M_0=m_0, M_1=m_1)\ \\ && \qquad \text{ under Assumption \ref{ass:nonlineardid}(\ref{ass2:nonlineardid})},\\
    &&= \mathbb P(U_0 \leq g_1^{-1}(0,m_1;y) \mid D=d, M_0=m_0, M_1=m_1)\ \\ && \qquad \text{ under Assumption \ref{ass:nonlineardid}(\ref{ass1:nonlineardid})}\\
    &&= \mathbb P(g_0(0,m_0,U_0) \leq g_0(0,m_0,g_1^{-1}(0,m_1;y)) \mid D=d, M_0=m_0, M_1=m_1) \ \\ && \qquad \text{ under Assumption \ref{ass:nonlineardid}(\ref{ass2:nonlineardid})},\\
    &&= \mathbb P(g_0(0,M_0,U_0) \leq T(m_0,m_1,y) \mid D=d, M_0=m_0, M_1=m_1),\\
    &&= \mathbb P(Y_0(0,M_0) \leq T(m_0,m_1,y) \mid D=d, M_0=m_0, M_1=m_1),\\
    &&= \mathbb P(Y_0 \leq T(m_0,m_1,y) \mid D=d, M_0=m_0, M_1=m_1),
\end{eqnarray*}
   where $T(m_0,m_1,y)\equiv g_0(0,m_0,g_1^{-1}(0,m_1;y))$. 

   For $d=0,$
\begin{eqnarray*}
   F_{Y_1 \mid D=0, M_0=m_0, M_1=m_1}(y) &\equiv& \mathbb P(Y_1 \leq y \mid D=0, M_0=m_0, M_1=m_1),\\
   &=& \mathbb P(Y_1(0,m_1) \leq y \mid D=0, M_0=m_0, M_1=m_1),\\
   &=& F_{Y_0 \mid D=0, M_0=m_0, M_1=m_1}(T(m_0,m_1,y))
\end{eqnarray*}
   Therefore, for $y \in \operatorname{supp}\{Y_1 \mid D=0, M_0=m_0, M_1=m_1\}$, we have $$T(m_0,m_1,y)= F^{-1}_{Y_0 \mid D=0, M_0=m_0, M_1=m_1}(F_{Y_1 \mid D=0, M_0=m_0, M_1=m_1}(y))$$ under Assumption \ref{ass:nonlineardid}(\ref{ass2:nonlineardid}).

   For $d=1$ and $y \in \operatorname{supp}\{Y_1(0, m_1) \mid D=1, M_0=m_0, M_1=m_1\}\subseteq \operatorname{supp}\{Y_1 \mid D=0, M_0=m_0, M_1=m_1\}$ where the inclusion follows from Assumption \ref{ass:nonlineardid}(\ref{ass3:nonlineardid}) and $Y_1(0, m_1) = g_1(0, m_1, U_1)$, we have
   \begin{eqnarray*}
  && \mathbb P(Y_1(0,m_1) \leq y \mid D=1, M_0=m_0, M_1=m_1) =  \\
       && \qquad \qquad F_{Y_0 \mid D=1, M_0=m_0, M_1=m_1}\left(F^{-1}_{Y_0 \mid D=0, M_0=m_0, M_1=m_1}(F_{Y_1 \mid D=0, M_0=m_0, M_1=m_1}(y))\right).
\end{eqnarray*}
To complete the proof, it is straightforward that 
\begin{eqnarray*}
      \mathbb P(Y_1 \leq y \mid D=1, M_0=m_0, M_1=m_1)  &=& \mathbb P(Y_1(1,m_1) \leq y \mid D=1, M_0=m_0, M_1=m_1). 
    \end{eqnarray*}
\end{proof}

\subsection{Proof of Proposition \ref{prop:directeffect_multiple}}\label{pf_prop:direffectmultiple}
\begin{proof}
We have
    \begin{eqnarray*}
    &&\mathbb E[\Delta Y_{it} \mid D_{it}=1, M_{it}=m_t, M_{i0}=m_0]- \mathbb E[\Delta Y_{it} \mid D_{it}=0, M_{it}=m_t, M_{i0}=m_0],\\
    &=& \mathbb E[\theta_{it} D_{it} +\beta_{it}M_{it}-\beta_{i0}M_{i0}+\eta_t-\eta_0 + U_{it}-U_{i0}\mid D_{it}=1, M_{it}=m_t, M_{i0}=m_0]\\
    && - \mathbb E[\theta_{it} D_{it} +\beta_{it}M_{it}-\beta_{i0}M_{i0}+\eta_t-\eta_0 + U_{it}-U_{i0}\mid D_{it}=0, M_{it}=m_t, M_{i0}=m_0],\\
    &=& \mathbb E[\theta_{it} \mid D_{it}=1, M_{it}=m_t, M_{i0}=m_0]\\
    &&+\mathbb E[\beta_{it}m_{t}-\beta_{i0}m_{0}+\eta_t-\eta_0 + U_{it}-U_{i0}\mid \\ && \qquad \qquad h_t(\lambda_i, \gamma_i,\varepsilon_{it})=1, \psi_t(1,\gamma_i,\tilde{V}_{it})=m_t, \psi_0(0,\gamma_i,\tilde{V}_{i0})=m_0]\\
    && - \mathbb E[\beta_{it}m_{t}-\beta_{i0}m_{0}+\eta_t-\eta_0 + U_{it}-U_{i0}\mid \\ && \qquad \qquad h_t(\lambda_i, \gamma_i,\varepsilon_{it})=0, \psi_t(0,\gamma_i,\tilde{V}_{it})=m_t, \psi_0(0,\gamma_i,\tilde{V}_{i0})=m_0],\\
    &=&\mathbb E[\theta_{it} \mid D_{it}=1, M_{it}=m_t, M_{i0}=m_0]\\
    &&+\mathbb E[\beta_{it}m_{t}-\beta_{i0}m_{0}+\eta_t-\eta_0 + U_{it}-U_{i0}]\\
    && - \mathbb E[\beta_{it}m_{t}-\beta_{i0}m_{0}+\eta_t-\eta_0 + U_{it}-U_{i0}],\\
    && \hspace{2cm}\text{ because } (U_{i0}, U_{it}, \beta_{i0}, \beta_{it})\ \indep\ (\tilde{V}_{i0}, \tilde{V}_{it}, \lambda_i, \gamma_i,\varepsilon_{it}) \text{ under Assumption \ref{ass:multiplenonstaggered}},\\
    &=&\mathbb E[\theta_{it} \mid D_{it}=1, M_{it}=m_t, M_{i0}=m_0].
\end{eqnarray*}
The overlap condition ensures that the estimand is well-defined. Hence, by integrating this estimand with respect to the distribution of the joint vector $(M_{it},M_{i0})$ conditional on treated group $D_{it}=1$, we obtain from the law of iterated expectations
\begin{eqnarray*}
    \mathbb E[\theta_{it}\mid D_{it}=1] 
    &=&\int \mathbb E[\Delta Y_{it} \mid D_{it}=1, M_{it}=m_t, M_{i0}=m_0]\\
    && \qquad \qquad - \mathbb E[\Delta Y_{it} \mid D_{it}=0, M_{it}=m_t, M_{i0}=m_0] dF_{M_{it},M_{i0}\mid D_{it}=1}(m_t,m_0). \nonumber
\end{eqnarray*}
\end{proof}

\section{Numerical simulations}
\subsection{DGP in Figure~\ref{fig:simulation_example}} \label{apx:simulation_example}
    Under our model \eqref{seq1}, we consider the following data generating process. Individual fixed effects and treatment-effect heterogeneity are jointly distributed as
    $$
(\lambda_{i},\gamma_{i},\theta_{i},\alpha_{i})\sim\mathcal{N}\left(\left(\begin{array}{c}
1\\
1\\
0.6\\
1
\end{array}\right),\left(\begin{array}{cccc}
1 & 0.5 & 0.5 & 0\\
0.5 & 1 & 0.5 & 0\\
0.5 & 0.5 & 1 & 0.5\\
0 & 0 & 0.5 & 1
\end{array}\right)\right).
$$
It allows the treatment effect on the outcome, $\theta_i$, to be correlated with fixed effects and the treatment effect on the mediator, $\alpha_i$. Treatment is selected based on fixed effects according to
$D_{i}=\mathbbm{1}\{\lambda_{i}+\gamma_{i}>\varepsilon_{i}\}$ where $\varepsilon_{i}\sim\mathcal{N}(2,1)$ is an exogenous shock, so treatment is endogenous through unobserved heterogeneity. 

The effects of the mediator on the outcome may vary over time and are allowed to be serially correlated: $(\beta_{i0},\beta_{i1})\sim\mathcal{N}\left(\left(\begin{array}{c}
3\\
1
\end{array}\right),
 \left(\begin{array}{cc}
1 & 0.5\\
0.5 & 1
\end{array}\right)\right)$. 

Idiosyncratic disturbances in the outcome and mediator equations are given by $(U_{i0},U_{i1})\sim\mathcal{N}\left(\left(\begin{array}{c}
0\\
0
\end{array}\right),
\left(\begin{array}{cc}
1 & 0.5\\
0.5 & 1
\end{array}\right)\right), (V_{i0},V_{i1})\sim\mathcal{N}\left(\left(\begin{array}{c}
0\\
0
\end{array}\right),\left(\begin{array}{cc}
1 & 0.5\\
0.5 & 1
\end{array}\right)\right)$. Nonrandom time effects are set to 
$(\eta_{0},\eta_{1})=(0,1)$, $(\delta_0,\delta_1)=(0,1)$. 

\subsection{More details on Example \ref{ex:testimp}}\label{apx_ex1}

Units $i=1,\ldots,n=4{,}000$ are observed in periods
$t\in\{-T_0,\ldots,0,1\}$ with $T_0=5$ pre-treatment periods before the baseline period $t=0$ and a single post-treatment period, $T=1$. 
The data follow
\begin{align}
Y_{it} &= \theta_i D_{it} + \beta_{it} M_{it} + \lambda_i + \eta_t + U_{it},
\label{eq:sim-y}\\
M_{it} &= \alpha D_{it} + \gamma_i + \delta_t + V_{it},
\label{eq:sim-m}\\
D_i &= \mathbbm{1}\{a_0 + a_\lambda \lambda_i + a_\gamma \gamma_i
       + \varepsilon_i > 0\},
\label{eq:sim-d}
\end{align}
with $\alpha=0.5$, $\eta_t=0.10\,t$, $\delta_t=0.15\,t$, $a_0=-0.25$ and $a_\lambda=a_\gamma=0.60$ where $D_{it}=D_i\,\mathbbm{1}\{t=1\}$. 
The unit fixed effects $(\lambda_i,\gamma_i)$ are bivariate standard normal with correlation $0.5$, and $\varepsilon_i\sim\mathcal{N}(0,1)$. 
The direct effect is $\theta_i = 0.5 + 0.3\,\lambda_i + 0.3\,\xi_i$ with $\xi_i\sim\mathcal{N}(0,1)$, so $\theta_{it}=\theta_i$ for all $t$. 
The return to the mediator exhibits trends over time,
\begin{equation}
\beta_{it} = \bar\beta + \beta_g\,t + b_i + \omega_{it},
\qquad \bar\beta=1,\ \beta_g=0.2,
\label{eq:sim-beta}
\end{equation}
with $b_i\sim\mathcal{N}(0,0.2^2)$ and $\omega_{it}\sim\mathcal{N}(0,0.1^2)$, so $\mathbb E[\beta_{it}] = 1 + 0.2\,t$.
The contemporaneous errors are $U_{it}\sim\mathcal{N}(0,0.5^2)$ and $V_{it}\sim\mathcal{N}(0,1)$. 
All draws are independent across units and periods and of one another, apart from the stated correlation between $\lambda_i$ and $\gamma_i$. 
This is model~\eqref{eq:multiplenonstaggered} with $\psi_t(D,\gamma,\tilde V)=\alpha D+\gamma+\delta_t+V$, $\tilde V_{it}=V_{it}$, and $h_t(\lambda,\gamma,\varepsilon)=\mathbbm{1}\{t=1\}\, \mathbbm{1}\{a_0+a_\lambda\lambda+a_\gamma\gamma+\varepsilon>0\}$ with $\varepsilon_{it}=\varepsilon_i$.

Note that Assumption~\ref{ass:multiplenonstaggered} holds, and $\beta_{it}$ is independent of the mediator path, so $\mathbb E[\Delta Y_{it}^{t_0}\mid D_{i1}=0, M_{i1}, M_{it_0}]
= \mathbb E[\beta_{i1}]M_{i1} - \mathbb E[\beta_{it_0}]M_{it_0}
+ \eta_1 - \eta_{t_0}$ is exactly linear in $(M_{i1},M_{it_0})$ for each base period. Hence, the outcome regression is correctly specified base period by base period, and the target
$\mathbb E[\theta_{i1}\mid D_{i1}=1]=0.5+0.3\,\mathbb E[\lambda_i\mid D_i=1]=0.671$
is identified from every $t_0$, so the profile is exactly flat. 
Its realized value in the simulated sample is $0.662$ (the average of $\theta_i$ over the $1{,}759$ treated units).

Case A uses the design as stated. Case B replaces $U_{it}$ by
$U_{it}+\tau\,t\,D_i$ with $\tau=0.15$, a linear trend in the outcome unobservable of the treated units, which violates Assumption~\ref{ass:multiplenonstaggered}(i), because $U_{it}$ then depends on $D_i$ and hence on $(\lambda_i,\gamma_i,\varepsilon_i)$. Under this violation, the estimand that uses base period $t_0$ equals $\mathbb E[\theta_{i1}\mid D_{i1}=1]+\tau(1-t_0)$, so the bias is $0.90$ at $t_0=-5$ and $0.15$ at $t_0=0$.

\section{Estimation and inference}\label{sec:estimation}
Our estimand for the average direct treatment effect on the treated has the same algebraic form as an ATT identified by conditional DiD with conditioning vector \((M_{i0},M_{i1})\). 
The conditioning vector should not be interpreted as an ordinary vector of pre-treatment covariates because \(M_{i1}\) is post-treatment for treated units. 
The identification results in the main text justify conditioning on the observed mediator path $(M_0,M_1)$ because the path-conditional DiD estimand helps identify the average direct effect under our model assumptions. 
Once the target is written in this form, existing DiD estimators for conditional ATT can be applied with \((M_{i0},M_{i1})\) as the conditioning vector.
In particular, we use the doubly robust (DR) estimator proposed by \cite{sant2020doubly}.

Let 
$$
\begin{aligned}
p(m_0,m_1)
&\equiv\mathbb P(D=1\mid M_0=m_0,M_1=m_1),\\
\mu_{d,t}(m_0,m_1)
&\equiv\mathbb E[Y_t\mid D=d,M_0=m_0,M_1=m_1].
\end{aligned}
$$
We consider working models $p(m_0,m_1;\boldsymbol\kappa_p)$ and $\mu_{d,t}(m_0,m_1;\boldsymbol\kappa_{d,t})$, where the common symbol $\boldsymbol\kappa$ represents nuisance-parameter vectors and the subscripts identify the propensity-score and outcome-regression models.
With the estimated nuisance parameters, we have
$$
\widehat\tau^{DR}
=
\mathbb E_n\!\left[
\left\{\widehat\omega_1(D_i)-
\widehat\omega_0(D_i,M_{i0},M_{i1};\widehat{\boldsymbol\kappa}_p)\right\}
\left\{\Delta Y_i-
\mu_{0,\Delta}(M_{i0},M_{i1};
\widehat{\boldsymbol\kappa}_{0,0},
\widehat{\boldsymbol\kappa}_{0,1})\right\}
\right],
$$
where $\mathbb E_n[\cdot]\equiv n^{-1}\sum_{i=1}^n[\cdot]$, $\Delta Y_i\equiv Y_{i1}-Y_{i0}$, $\Delta M_i\equiv M_{i1}-M_{i0}$, and
\begin{align*}
\widehat\omega_1(D_i)
&=\frac{D_i}{\mathbb E_n[D_i]},\\
\widehat\omega_0(D_i,M_{i0},M_{i1};\boldsymbol\kappa_p)
&=
\frac{
p(M_{i0},M_{i1};\boldsymbol\kappa_p)(1-D_i)/[1-p(M_{i0},M_{i1};\boldsymbol\kappa_p)]
}{
\mathbb E_n\!\left[
p(M_{i0},M_{i1};\boldsymbol\kappa_p)(1-D_i)/[1-p(M_{i0},M_{i1};\boldsymbol\kappa_p)]
\right]
},\\
\mu_{0,\Delta}(M_{i0},M_{i1};\boldsymbol\kappa_{0,0},\boldsymbol\kappa_{0,1})
&=
\mu_{0,1}(M_{i0},M_{i1};\boldsymbol\kappa_{0,1})
-
\mu_{0,0}(M_{i0},M_{i1};\boldsymbol\kappa_{0,0}).
\end{align*}
If either the propensity-score model $p(m_0,m_1;\boldsymbol\kappa_p)$ or the outcome-regression model $\mu_{0,t}(m_0,m_1;\boldsymbol\kappa_{0,t})$ is correctly specified, this estimator is consistent for $\mathbb{E}\big[ \tau^{DiD}(M_{i0},M_{i1}) \big| D_i =1\big]$, which identifies the ADTT.
When the propensity-score parameters are estimated by inverse probability tilting
and the outcome-regression parameters by odds-weighted least squares on the control units, the estimator is doubly robust for inference as well as for consistency, so the variance estimator below remains valid when one of the two working models is misspecified, and it attains the semiparametric efficiency bound when both are correctly specified \citep{sant2020doubly}.

For the total treatment effect on outcomes (ATT) and the treatment effect on the mediator, we can use the standard DiD estimators $\widehat\tau^{DiD}$ and $\widehat\alpha^{DiD}$, the sample analogues of $\tau^{DiD}$ and $\alpha^{DiD}$ under the stated assumptions. 
Write $\bar\beta\equiv \mathbb E[\beta_{i0}] =\mathbb E[\beta_{i1}]$ for the average mediator effect of Proposition \ref{prop:medeffect} under Assumption \ref{ass:stablebeta}.
Note that the indirect-effect estimator $\widehat{AITT}=\widehat\tau^{DiD}-\widehat{ADTT}$, where $\widehat{ADTT}\equiv\widehat\tau^{DR}$ denotes the doubly robust estimator of the direct effect, is a linear transformation of jointly estimated components, and the ratio $\widehat{\bar\beta}=\widehat{AITT}/\widehat\alpha^{DiD}$ is a smooth transformation when the mediator-response denominator is bounded away from zero.

With exogenous baseline covariates $\boldsymbol X_i$, it can be shown that the direct effect, the total effect, and the mediator response are identified by the same doubly robust construction with the covariates stacked into the conditioning vector, under the same assumptions conditional on $\boldsymbol X_i$.
Let $\boldsymbol W_i$ represent the conditioning vector of a given target.
Then, we can use $\boldsymbol W_i=(M_{i0},M_{i1},\boldsymbol X_i')'$ for the direct effect, and $\boldsymbol W_i=\boldsymbol X_i$ for the total effect and for the treatment effect on the mediator. 
In this case, the working models $p(\boldsymbol W_i;\boldsymbol\kappa_p)$ and $\mu_{d,t}(\boldsymbol W_i;\boldsymbol\kappa_{d,t})$ are now indexed by the full vector $\boldsymbol W_i$, and the doubly robust estimator is unchanged in form. 
Table \ref{tab:dr-mapping} summarizes the mapping between estimands, conditioning vectors, and identified parameters.

\begin{table}[!htbp]
\centering
\begin{threeparttable}
\caption{Estimands, conditioning vectors, and identified parameters}
\label{tab:dr-mapping}
\renewcommand{\arraystretch}{1.3}
\setlength{\tabcolsep}{8pt}
\begin{tabular}{llll}
\toprule
Outcome & Conditioning vector $\boldsymbol W_i$ & Identified parameter & Result \\
\midrule
$\Delta Y_i$ & $(M_{i0},\,M_{i1},\,\boldsymbol X_i')'$ & $ADTT$ & Propositions \ref{prop:direffect}, \ref{prop:generalresultADTT} \\
$\Delta Y_i$ & $\boldsymbol X_i$ & $ATT$ & Proposition \ref{prop:generalresultAITT} \\
$\Delta M_i$ & $\boldsymbol X_i$ & $\mathbb E[\alpha_i]$ & Proposition \ref{prop:medeffect} \\
\midrule
\multicolumn{4}{l}{Derived quantities:\quad $AITT=ATT-ADTT$ (Proposition \ref{prop:generalresultAITT});}\\
\multicolumn{4}{l}{\phantom{Derived quantities:}\quad $\mathbb E[\beta_{i1}]=AITT/\mathbb E[\alpha_i]$ (Proposition \ref{prop:medeffect})}\\
\bottomrule
\end{tabular}
  \par\footnotesize
  \emph{Notes.} Each row applies the doubly robust estimator of \cite{sant2020doubly} to the indicated outcome ($\Delta Y_i\equiv Y_{i1}-Y_{i0}$ and $\Delta M_i\equiv M_{i1}-M_{i0}$) with the indicated conditioning vector entering both the propensity-score and outcome-regression working models. Without baseline covariates, the first row applies with $\boldsymbol W_i=(M_{i0},M_{i1})'$ and the remaining rows reduce to unadjusted DiD contrasts. 
\end{threeparttable}
\end{table}

For the three improved doubly robust fits, we stack the estimating equations into one just-identified GMM system.
Fix a target with conditioning vector $\boldsymbol W_i$ and outcome evolution $\Delta Y_i$ or $\Delta M_i$ for the three fits of Table \ref{tab:dr-mapping}, and set $\boldsymbol x_i=(1,\boldsymbol W_i')'$.
Parameterize the propensity-score working model through the treatment odds $r_i(\boldsymbol\kappa_p)=\exp(\boldsymbol x_i'\boldsymbol\kappa_p)$, so that $p(\boldsymbol W_i;\boldsymbol\kappa_p)=r_i(\boldsymbol\kappa_p)/\{1+r_i(\boldsymbol\kappa_p)\}$, and take the control outcome regression to be linear, $\mu_i(\boldsymbol\kappa_{or})=\boldsymbol x_i'\boldsymbol\kappa_{or}$, a direct working model for the control evolution $\mu_{0,\Delta}$.
Write $q_i(\boldsymbol\kappa_p)=D_i-(1-D_i)r_i(\boldsymbol\kappa_p)$ and $u_i(\boldsymbol\kappa_{or})=\Delta Y_i-\mu_i(\boldsymbol\kappa_{or})$ or $u_i(\boldsymbol\kappa_{or})=\Delta M_i-\mu_i(\boldsymbol\kappa_{or})$, and let
\[
\mathbbm 1^{\mathrm{tr}}_i(\boldsymbol\kappa_p)=\mathbbm 1\{D_i=1\}+\mathbbm 1\{D_i=0,\ p(\boldsymbol W_i;\boldsymbol\kappa_p)\le\bar p\},\qquad \bar p=0.995,
\]
be the control trimming indicator, which removes control units whose fitted propensity score is close to one.

With $\boldsymbol\eta=(\boldsymbol\kappa_p',\boldsymbol\kappa_{or}',\theta)'$, the target-specific moment vector is
\begin{equation}\label{eq:dr-moments}
\boldsymbol g_i(\boldsymbol\eta)=
\begin{pmatrix}
q_i(\boldsymbol\kappa_p)\,\boldsymbol x_i\\[0.12cm]
(1-D_i)\,r_i(\boldsymbol\kappa_p)\,u_i(\boldsymbol\kappa_{or})\,\boldsymbol x_i\\[0.12cm]
\mathbbm 1^{\mathrm{tr}}_i(\boldsymbol\kappa_p)\,q_i(\boldsymbol\kappa_p)\,u_i(\boldsymbol\kappa_{or})-D_i\theta
\end{pmatrix}.
\end{equation}
The first block is the exact calibration condition of inverse probability tilting,
the second block is the odds-weighted least squares fit of the control outcome evolution, and the third block defines the estimand.
These two nuisance fits are what make the estimator of \cite{sant2020doubly} the improved one.

The intercept row of the first block gives $\mathbb E_n[(1-D_i)r_i]=\mathbb E_n[D_i]$.
Hence the single normalization by $\mathbb E_n[D_i]$ in the third block reproduces the separately normalized weights $\widehat\omega_1$ and $\widehat\omega_0$ defined above, and the estimand coordinate $\widehat\theta$ of the solution $\widehat{\boldsymbol\eta}$ of $\sum_{i=1}^n\boldsymbol g_i(\boldsymbol\eta)=\boldsymbol 0$ coincides with $\widehat\tau^{DR}$.
This equivalence uses the calibration condition, so it requires the tilting fit rather than a maximum likelihood logit fit, and it requires the trimming rule to be inactive.\footnote{In the railroads sample the largest control propensity score is $0.994$, below the $0.995$ threshold, so no control unit is trimmed and the identity holds exactly.}
When the trimming rule binds, $\widehat\theta$ estimates the trimmed analogue where the outer average is taken over the retained support, and the equivalence with $\widehat\tau^{DR}$ holds only if the same trimmed score and the same normalization are used in both constructions.

The full system stacks \eqref{eq:dr-moments} across the three targets in Table \ref{tab:dr-mapping}.
Note that the total-effect and mediator-response estimators use the same conditioning vector \(\boldsymbol x_i^{\mathrm{tot}}=(1,\boldsymbol X_i')'\), and therefore share a common propensity-score tilting equation, while retaining separate outcome-regression equations for \(\Delta Y_i\) and \(\Delta M_i\), respectively. 
The direct-effect estimator instead uses \(\boldsymbol x_i^{\mathrm{dir}}=(1,M_{i0},M_{i1},\boldsymbol X_i')'\), and therefore has its own propensity-score tilting and outcome-regression equations.
With multiple pre-treatment periods, the base-period profile of Remark~\ref{rem:pretreatment} fits into the same system.
Each base period $t_0\in\{-T_0,\ldots,0\}$ contributes one additional direct-effect block \eqref{eq:dr-moments} with outcome evolution $\Delta Y_{i1}^{t_0}\equiv Y_{i1}-Y_{it_0}$ and its own tilt and regression on $\boldsymbol x^{t_0}_i=(1,M_{it_0},M_{i1},\boldsymbol X_i')'$, the $t_0=0$ block coinciding with the direct-effect fit above.
Collecting all nuisance and estimand coordinates in $\boldsymbol\eta$, with $\boldsymbol g_i(\boldsymbol\eta)$ now denoting the stacked moment vector, the estimator solves $\sum_{i=1}^n\boldsymbol g_i(\widehat{\boldsymbol\eta})=\boldsymbol 0$.

Let $A=\mathbb E[\partial\boldsymbol g_i(\boldsymbol\eta_0)/\partial\boldsymbol\eta']$ at the probability limit $\boldsymbol\eta_0$, taking $\mathbbm 1^{\mathrm{tr}}_i$ as fixed, which is valid when trimming is inactive with probability approaching one.
Under standard regularity conditions for just-identified GMM estimators, $\widehat{\boldsymbol\eta}$ is asymptotically linear with influence function $-A^{-1}\boldsymbol g_i(\boldsymbol\eta_0)$, and its component-specific pieces reproduce the improved doubly robust influence functions of \cite{sant2020doubly} \citep{newey1994large}.
With $\widehat{\boldsymbol G}_c=\sum_{i\in c}\boldsymbol g_i(\widehat{\boldsymbol\eta})$ over clusters $c=1,\ldots,C$,
\[
\widehat A=n^{-1}\sum_{i=1}^n\frac{\partial\boldsymbol g_i(\widehat{\boldsymbol\eta})}{\partial\boldsymbol\eta'},
\qquad
\widehat\Omega_{\mathrm{cl}}=n^{-1}\sum_{c=1}^C\widehat{\boldsymbol G}_c\widehat{\boldsymbol G}_c',
\qquad
\widehat V=n^{-1}\widehat A^{-1}\widehat\Omega_{\mathrm{cl}}(\widehat A^{-1})'.
\]
This construction retains all cross-covariances among $\widehat\tau^{DiD}$, $\widehat{ADTT}$, and $\widehat\alpha^{DiD}$, after which the delta method gives inference for $\widehat{AITT}$ and $\widehat{\bar\beta}$. 
With multiple pre-treatment periods, the estimand block of $\widehat V$ likewise gives the joint covariance matrix of the base-period estimates $\hat{\bm b} \equiv (\widehat\theta(-T_0),\ldots,\widehat\theta(0))'$, on which the joint Wald and sup-$t$ tests of the constancy implication of Remark~\ref{rem:pretreatment} are based (Figure~\ref{fig:pretest-sim}).

In particular, consider a restriction matrix $R$ of full row rank with $J$ rows and a given null hypothesis $R \bm b = \bm 0$, the joint Wald statistic is constructed as $W=(R\hat{\bm b})'(R\hat V R')^{-1}(R\hat{\bm b})$ which is asymptotically $\chi^2$ distributed with $J$ degrees of freedom under the null.\footnote{Note that $W$ does not change for any non-singular $J\times J$ matrix $A$, so the choice of reference period can be arbitrary for testing the constancy in Figure~\ref{fig:pretest-sim}.}
The sup-$t$ test is based on $T=\max_j|(R\hat {\bm b})_j|/\sigma_j$ where $\sigma_j^2=(R\hat V R')_{jj}$, and the $p$-value is calibrated by Gaussian draws $\bm \xi_b\sim\mathcal{N}(\bm 0,R\hat V R')$ for $b= 1, \ldots, B$ as $p=(1+\#\{b\colon T_b\ge T\})/(B+1)$, where $T_b=\max_j|\xi_{bj}|/\sigma_j$.
Note that the simultaneous $100(1-\alpha)\%$ band can be obtained as $(R\hat {\bm b})_j\pm c\,\sigma_j$ where $c$ is the $1-\alpha$ quantile of $T_b$ over the same draws \citep[Algorithm 1]{montiel2019simultaneous}, so the sup-$t$ test rejects at the $100\alpha\%$ level exactly when one of the $J$ intervals of the simultaneous $100(1-\alpha)\%$ band excludes zero, up to the discreteness of the $B$ draws.\footnote{When $J=1$, $W=T^2$, the simultaneous band reduces to the pointwise interval, and the two tests coincide up to the simulation error in the sup-$t$ $p$-value.

}

\section{Additional empirical results}\label{apx:empirical}

This Appendix reports a second application, to Alcohol Prohibition in the early twentieth-century United States, and the no-covariate estimates for the railroads application of Section \ref{sec:empirical}.

\subsection{Alcohol Prohibition and farm productivity}\label{apx:prohibition}

\cite{howard2021closing} compare counties adopting Alcohol Prohibition during 1900--1910 with counties adopting after 1910. Their DiD estimates show higher farm productivity and land values in early-adopting counties, and they interpret the land-value response as evidence that Prohibition relaxed borrowing constraints and supported capital investment.

We retain this treatment comparison but set $Y_{it}$ to log farm productivity and $M_{it}$ to log farm value per acre. 
The county panel of \citet{howard2021closing} carries three census years, $1890$, $1900$ and $1910$, but we keep $1900$ and $1910$ and treat $1900$ as the baseline for the main illustration (Table~\ref{tab:estimates}) and additionally utilize the year of $1890$ for the illustration of Remark~\ref{rem:pretreatment} (Figure~\ref{fig:ho-pretest}).\footnote{The full two-period sample has $2{,}342$ counties, with $1{,}316$ treated and $1{,}026$ controls, and the complete-covariate sample has $2{,}301$ counties, with $1{,}296$ treated and $1{,}005$ controls.
On the other hand, we remove counties whose recorded adoption year is $1900$ itself for the three-period sample with non-missing variables for all three periods; the resulting sample contains $2{,}260$ counties, of which $1{,}256$ are treated and $1{,}004$ are controls.}
Baseline controls consist of the 1890 share-in-favor decile indicators and log population plus the urban, male, white, foreign, and foreign-born shares. Here the direct effect is the productivity component not transmitted through the observed farm-value path, so labor-efficiency and other unmeasured channels remain in that component.

Table \ref{tab:estimates} compares the descriptive TWFE decomposition with the improved DR estimator of \cite{sant2020doubly}. 
The linear DR family uses the original baseline-control span for all three targets and adds $(M_{i0},M_{i1})$ only to the direct-effect fit. The additive-quadratic family adds squares of the six continuous controls to all three fits and $(M_{i0}^2,M_{i1}^2)$ to the direct-effect nuisance model, but no $M_{i0}M_{i1}$ interaction. Standard errors come from a joint county-clustered CR0 stacked-sandwich covariance without a finite-cluster degrees-of-freedom adjustment. The indirect effect and mediator-effect ratio use the delta method.

\input{HO_output/tab_estimates_ho}

Column (1) of Table \ref{tab:estimates} gives the descriptive TWFE decomposition on the complete-covariate sample. The total productivity coefficient is $0.094$ $(0.038)$, the mediator-adjusted Prohibition coefficient is $0.065$ $(0.038)$, and the mediator response is $0.057$ $(0.024)$. These values reproduce the positive land-value channel in the regression decomposition, but Claim \ref{claim:interpretTWFE} rules out assigning causal direct and indirect interpretations without the paper's identifying restrictions.

Under those restrictions, the linear DR fit gives $ATT=0.133$ $(0.041)$, $ADTT=0.103$ $(0.046)$, and a mediator response of $0.135$ $(0.028)$. The implied $AITT$ is $0.030$ $(0.019)$ and the mediator-effect ratio is $0.218$ $(0.121)$. The $ADTT$ differs from zero at the 5\% level, the $AITT$ does not, and the ratio differs from zero only at the 10\% level. Adding the quadratic terms changes the $ADTT$ by only $0.003$, to $0.106$ $(0.047)$, but reduces the $AITT$ to $0.009$ $(0.020)$ and the ratio to $0.063$ $(0.135)$. The direct estimate is therefore insensitive to these added terms, whereas evidence for mediation through observed farm value is not.

\clearpage
\input{HO_output/tab_estimates_ho_complete}

On the other hand, Table \ref{tab:estimates_complete} provides the whole picture of the estimates across different specifications along with the use of covariates. Without baseline controls, the TWFE approach gives a total-effect estimate of $0.009$ $(0.033)$ and a mediator response of $-0.115$ $(0.026)$. The linear and additive-quadratic DR direct estimates are $0.080$ $(0.027)$ and $0.150$ $(0.028)$, with the indirect effects calculated as $-0.071$ $(0.017)$ and $-0.141$ $(0.023)$. With baseline controls, the total and mediator-response estimates are positive, and the indirect estimates are close to zero.

Lastly, we utilize the three-period sample to estimate the ADTT for each base period $t_0\in\{1890,1900\}$ using the improved doubly robust DiD estimator applied to the two-period panel $(t_0,1910)$, which we denote by $\hat\tau^{DiD}_{t_0}$.
As the two estimates use the same counties and share $Y_{i,1910}$ creating a strong correlation, we obtain the joint covariance of the two estimators for the Wald statistic by stacking the moment conditions of both base periods into one just-identified GMM system.
It is notable that unconditional parallel trends fails over the two pre-treatment periods for both variables in the constructed three-period sample; the two-group gap in $Y_{i,1900}-Y_{i,1890}$ is $-0.246$ with a county-robust standard error of $0.029$, and the gap in $M_{i,1900}-M_{i,1890}$ is $-0.210$ with a standard error of $0.024$.

The estimated ADTTs at $1910$ and their difference with the joint Wald test are summarized in Figure~\ref{fig:ho-pretest}
for each combination of linear / quadratic mediator path and with / without the baseline controls.
Note that the two ADTT estimates disagree sharply when we do not include the baseline covariates while they are statistically indistinguishable from each other with the controls, suggesting that the model along with the assumptions and the DR specifications is rejected at the $5\%$ significance level without the baseline covariates.

\begin{figure}[!htbp]
    \centering
    \includegraphics[width=\linewidth]{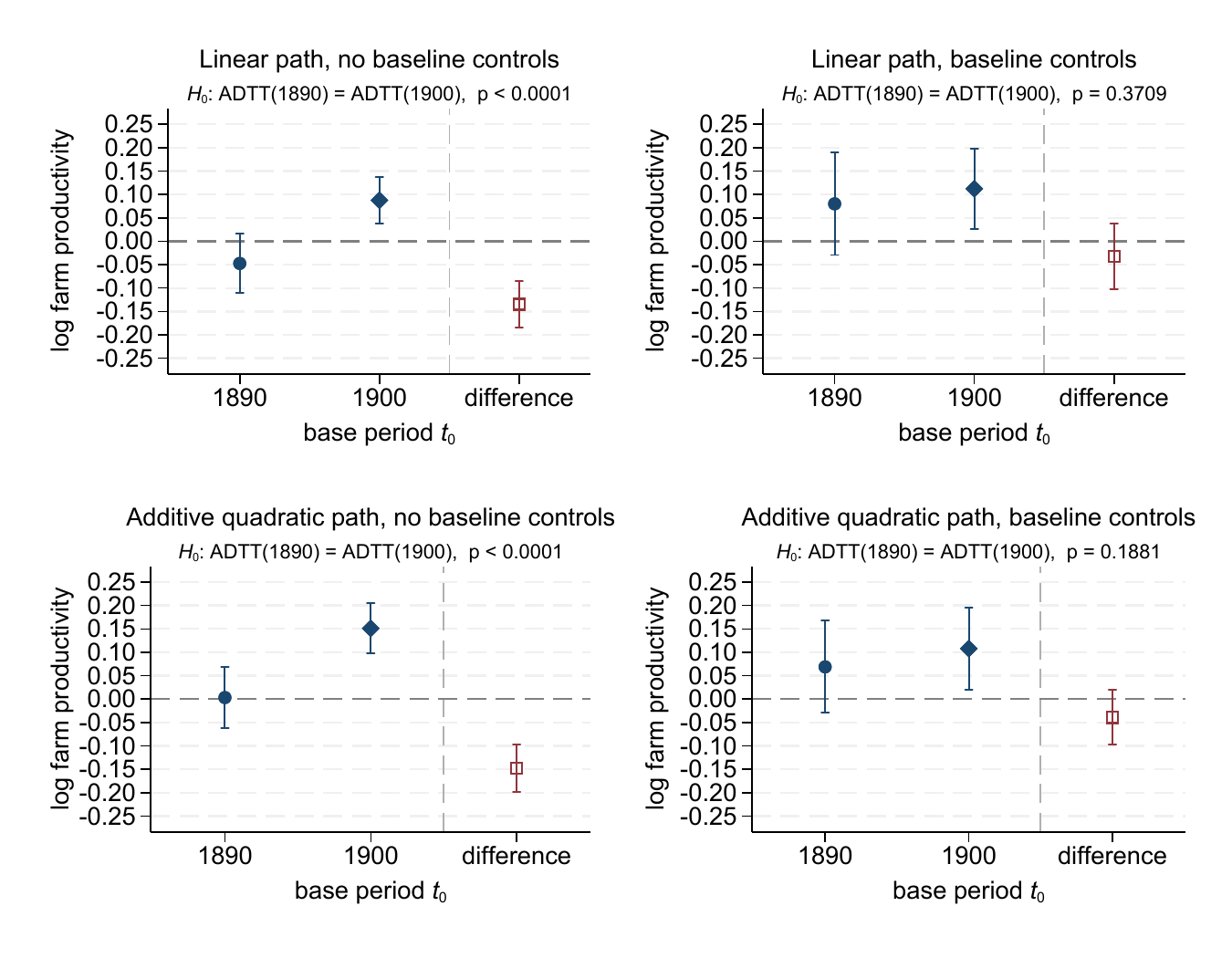}
    \caption{Base-period profile of the direct effect at $1910$, Prohibition}
    \label{fig:ho-pretest}
    \begin{minipage}{\linewidth}
    \footnotesize
\emph{Notes.} Left of the divider in each panel, the circle is $\hat\tau^{DiD}_{1890}$ and the diamond is $\hat\tau^{DiD}_{1900}$ with pointwise $95\%$ county-clustered Wald intervals.
Right of the divider, the hollow square is the difference $\hat\tau^{DiD}_{1890}-\hat\tau^{DiD}_{1900}$ with its own pointwise $95\%$ interval, built from the joint covariance of the two estimates.
The dashed horizontal line represents zero.
The subtitles report the joint Wald $p$-value for $H_0\colon\mathrm{ADTT}(1890)=\mathrm{ADTT}(1900)$. 
    \end{minipage}
\end{figure}

\newpage

\subsection{Railroads application without baseline covariates}\label{apx:dh-nocov}

We repeat the mediation analysis of the local-adoption analysis with no baseline covariates and summarize the results in Table \ref{tab:estimates_dh_nocov}. Note that the total effects and the effects on the mediator are obtained as unadjusted DiD estimands of $0.520$ for the outcome and $0.756$ for the mediator across all columns. Without the baseline covariates, we obtain $ADTT=0.547$, $AITT=-0.027$, and a mediator-effect ratio of $-0.036$ with the linear mediator path, and $ADTT=0.599$, $AITT = -0.079$, and a mediator-effect ratio of $-0.105$ with the additive-quadratic path.

In the no-covariate TWFE regression, in contrast, the direct effect estimate is obtained as $0.086$, and the claimed indirect effect estimate is $0.434$. It has to be noted that Claim \ref{claim:interpretTWFE} still limits the TWFE estimates to a descriptive decomposition.

\input{DH_output/tab_estimates_dh_nocov}

\end{document}

%% file: DH_output/tab_estimates_dh.tex
\begin{table}[!htbp]
\centering
\begin{threeparttable}
\footnotesize
\caption{Local railroad adoption and market-access mediation, 1870--1890}
\label{tab:estimates_dh}
\renewcommand{\arraystretch}{1.25}
\setlength{\tabcolsep}{3pt}
\begin{tabular}{lcccc}
\toprule
& \multicolumn{2}{c}{TWFE} & \multicolumn{2}{c}{Doubly Robust} \\
\cmidrule(lr){2-3}\cmidrule(lr){4-5}
& Original controls & Same controls & Linear path & Additive quadratic \\
& (1) & (2) & (3) & (4) \\
\midrule
$D \to Y$ (Total) & -0.109 & 0.348\sym{**} & 0.390\sym{**} & 0.390\sym{**} \\
 & (0.087) & (0.147) & (0.158) & (0.158) \\[3pt]
$D \to Y$ (Direct) & -0.253\sym{***} & 0.090 & 0.429\sym{**} & 0.525\sym{*} \\
 & (0.098) & (0.175) & (0.177) & (0.275) \\[3pt]
$D \to M$ & 0.376\sym{***} & 0.656\sym{***} & 0.698\sym{***} & 0.698\sym{***} \\
 & (0.054) & (0.100) & (0.107) & (0.107) \\[3pt]
$M \to Y$ & 0.384\sym{***} & 0.394\sym{**} & -0.055 & -0.193 \\
 & (0.114) & (0.178) & (0.204) & (0.386) \\[3pt]
$D \to M \to Y$ (Indirect) & 0.144\sym{***} & 0.258\sym{**} & -0.039 & -0.135 \\
 & (0.043) & (0.118) & (0.141) & (0.262) \\
\midrule
Covariate adjustment & State + cubic geo. & Linear geo. & Linear geo. & Linear geo. \\
Observations & 2{,}110 & 2{,}110 & 2{,}110 & 2{,}110 \\
\bottomrule
\end{tabular}
\par\footnotesize
\textit{Notes}: Columns (1) and (2) are descriptive TWFE decompositions. Their rows report, in order, the total local-adoption effect coefficient from the outcome regression without $M$, that coefficient after adding $M$, targeting the direct effect, the local-adoption effect coefficient from the mediator regression, the coefficient on $M$ in the outcome regression, and the difference between the first two coefficients.
The specification of Column (1) uses state-by-period effects and separate cubic terms in longitude and latitude. Column (2) uses the same $1{,}055$ counties, equal observation weights, and linear geographic covariates as Column (3). Columns (3) and (4) report DR estimates of $ATT$, $ADTT$, the treatment effect on $M$, the mediator-effect ratio, and $AITT$ under our framework.
The sample consists of $730$ adopters and $325$ controls observed in two periods. Standard errors use a joint state-clustered CR0 stacked-sandwich covariance across $37$ states, and the last two rows use the delta method.
$^{*}$, $^{**}$, and $^{***}$ denote statistical significance at the 10\%, 5\%, and 1\% levels , respectively.
\end{threeparttable}
\end{table}

%% file: HO_output/tab_estimates_ho.tex
\begin{table}[!htbp]
\centering
\begin{threeparttable}
\caption{Total, direct, and indirect treatment effects: Prohibition}
\label{tab:estimates}
\renewcommand{\arraystretch}{1.3}
\setlength{\tabcolsep}{9pt}
\begin{tabular}{lccc}
\toprule
& \multirow{2}{*}{TWFE} & \multicolumn{2}{c}{Doubly robust} \\
\cmidrule(lr){3-4}
& & Linear path & Additive quadratic \\
& (1) & (2) & (3) \\
\midrule
$D \to Y$ (Total) & 0.094\sym{**} & 0.133\sym{***} & 0.115\sym{***} \\
 & (0.038) & (0.041) & (0.041) \\[3pt]
$D \to Y$ (Direct) & 0.065\sym{*} & 0.103\sym{**} & 0.106\sym{**} \\
 & (0.038) & (0.046) & (0.047) \\[3pt]
$D \to M$ & 0.057\sym{**} & 0.135\sym{***} & 0.141\sym{***} \\
 & (0.024) & (0.028) & (0.027) \\[3pt]
$M \to Y$ (Ratio) & 0.498\sym{***} & 0.218\sym{*} & 0.063 \\
 & (0.081) & (0.121) & (0.135) \\[3pt]
$D \to M \to Y$ (Indirect) & 0.028\sym{**} & 0.030 & 0.009 \\
 & (0.012) & (0.019) & (0.020) \\
\midrule
Baseline controls & Original span & Original span & Original + squares \\
Observations & 4{,}602 & 4{,}602 & 4{,}602 \\
\bottomrule
\end{tabular}
\par\footnotesize
\textit{Notes}: Column (1) is the descriptive TWFE decomposition. Columns (2) and (3) report DR estimates of $ATT$, $ADTT$, the treatment effect on $M$, the mediator-effect ratio, and $AITT$ under our framework. The outcome is the logarithm of farm productivity, the mediator is the logarithm of farm value per acre, and all columns use the complete-covariate sample. The baseline span contains the 1890 share-in-favor decile indicators and six continuous 1890 controls. Column (3) adds squares of those six controls, and the direct-effect specification also contains $M_{i0}^2$ and $M_{i1}^2$ without $M_{i0}M_{i1}$.
Standard errors are obtained using a joint county-clustered CR0 stacked-sandwich covariance across $2{,}301$ counties, and the last two rows use the delta method.
$^{*}$, $^{**}$, and $^{***}$ denote significance at the 10\%, 5\%, and 1\% levels , respectively.
\end{threeparttable}
\end{table}

%% file: HO_output/tab_estimates_ho_complete.tex
\begin{landscape}
\begin{table}[!p]
\centering
\begin{threeparttable}
\caption{Total, direct, and indirect treatment effects: Prohibition, with and without baseline controls}
\label{tab:estimates_complete}
\setlength{\tabcolsep}{7pt}
\renewcommand{\arraystretch}{1.4}
\begin{tabular}{lcccccc}
\toprule
& \multicolumn{2}{c}{TWFE} & \multicolumn{2}{c}{DR: linear path} & \multicolumn{2}{c}{DR: additive quadratic} \\
\cmidrule(lr){2-3}\cmidrule(lr){4-5}\cmidrule(lr){6-7}
& No controls & Controls & No controls & Controls & No controls & Controls \\
& (1) & (2) & (3) & (4) & (5) & (6) \\
\midrule
$D \to Y$ (Total) & 0.009 & 0.094\sym{**} & 0.009 & 0.133\sym{***} & 0.009 & 0.115\sym{***} \\
 & (0.033) & (0.038) & (0.033) & (0.041) & (0.033) & (0.041) \\[3pt]
$D \to Y$ (Direct) & 0.079\sym{***} & 0.065\sym{*} & 0.080\sym{***} & 0.103\sym{**} & 0.150\sym{***} & 0.106\sym{**} \\
 & (0.029) & (0.038) & (0.027) & (0.046) & (0.028) & (0.047) \\[3pt]
$D \to M$ & -0.115\sym{***} & 0.057\sym{**} & -0.115\sym{***} & 0.135\sym{***} & -0.115\sym{***} & 0.141\sym{***} \\
 & (0.026) & (0.024) & (0.026) & (0.028) & (0.026) & (0.027) \\[3pt]
$M \to Y$ (Ratio) & 0.601\sym{***} & 0.498\sym{***} & 0.616\sym{***} & 0.218\sym{*} & 1.222\sym{***} & 0.063 \\
 & (0.053) & (0.081) & (0.053) & (0.121) & (0.209) & (0.135) \\[3pt]
$D \to M \to Y$ (Indirect) & -0.069\sym{***} & 0.028\sym{**} & -0.071\sym{***} & 0.030 & -0.141\sym{***} & 0.009 \\
 & (0.018) & (0.012) & (0.017) & (0.019) & (0.023) & (0.020) \\
\midrule
Observations & 4{,}684 & 4{,}602 & 4{,}684 & 4{,}602 & 4{,}684 & 4{,}602 \\
\bottomrule
\end{tabular}
\par\footnotesize
\textit{Notes}: Odd-numbered columns omit baseline controls, and even-numbered columns include them. The TWFE columns are descriptive decompositions, and the DR columns report DR estimates of $ATT$, $ADTT$, the treatment effect on $M$, the mediator-effect ratio, and $AITT$ under our framework.
The linear direct-effect specifications condition on $(M_{i0},M_{i1})$, and the additive-quadratic models additionally condition on $(M_{i0}^2,M_{i1}^2)$ without their interaction.
Standard errors are obtained using a joint county-clustered CR0 stacked-sandwich covariance, with $2{,}342$ county clusters in the uncontrolled columns and $2{,}301$ in the controlled columns. The last two rows use the delta method.
\end{threeparttable}
\end{table}
\end{landscape}

%% file: DH_output/tab_estimates_dh_nocov.tex
\begin{table}[!htbp]
\centering
\begin{threeparttable}
\small
\caption{Local railroad adoption and market-access mediation, 1870--1890 (no baseline covariates)}
\label{tab:estimates_dh_nocov}
\renewcommand{\arraystretch}{1.25}
\setlength{\tabcolsep}{5pt}
\begin{tabular}{lccc}
\toprule
& \multirow{2}{*}{TWFE} & \multicolumn{2}{c}{Doubly Robust} \\
\cmidrule(lr){3-4}
& & Linear path & Additive quadratic \\
& (1) & (2) & (3) \\
\midrule
$D \to Y$ (Total) & 0.520\sym{***} & 0.520\sym{***} & 0.520\sym{***} \\
 & (0.173) & (0.173) & (0.173) \\[3pt]
$D \to Y$ (Direct) & 0.086 & 0.547\sym{***} & 0.599\sym{***} \\
 & (0.190) & (0.183) & (0.206) \\[3pt]
$D \to M$ & 0.756\sym{***} & 0.756\sym{***} & 0.756\sym{***} \\
 & (0.114) & (0.114) & (0.114) \\[3pt]
$M \to Y$ & 0.574\sym{***} & -0.036 & -0.105 \\
 & (0.163) & (0.174) & (0.228) \\[3pt]
$D \to M \to Y$ (Indirect) & 0.434\sym{***} & -0.027 & -0.079 \\
 & (0.118) & (0.130) & (0.168) \\
\midrule
Covariate adjustment & None & None & None \\
Observations & 2{,}110 & 2{,}110 & 2{,}110 \\
\bottomrule
\end{tabular}
\par\footnotesize
\textit{Notes}: Column (1) is the descriptive TWFE decomposition without baseline covariates. Columns (2) and (3) report DR estimates of $ATT$, $ADTT$, the treatment effect on $M$, the mediator-effect ratio, and $AITT$ under our framework.
The outcome is the log value of agricultural land and buildings, and the mediator is the log of cost-based market access. The linear direct-effect specification conditions on $(M_{o0},M_{o1})$, and the additive-quadratic model additionally conditions on $(M_{o0}^2,M_{o1}^2)$ without their interaction. The sample contains $730$ adopters and $325$ controls observed in two periods. Standard errors use a joint state-clustered CR0 stacked-sandwich covariance across $37$ states, and the last two rows use the delta method.
$^{*}$, $^{**}$, and $^{***}$ denote statistical significance at the 10\%, 5\%, and 1\% levels , respectively.
\end{threeparttable}
\end{table}